\documentclass[twocolumn,10pt]{IEEEtran}
\usepackage{comment}
\usepackage[linesnumbered,ruled,vlined]{algorithm2e}
\usepackage{mathtools, cases}
\usepackage{textcomp}
\includecomment{toexclude}
\usepackage{amsthm, enumitem}
\usepackage{amsmath,amsfonts,amssymb}
\usepackage{cite, url}
\usepackage{verbatim}
\usepackage{bm}
\usepackage{graphicx}
\usepackage{xcolor}
\usepackage{subcaption}
\usepackage{epstopdf}
\usepackage{mathrsfs}
\usepackage{txfonts}
\usepackage{lineno}
\usepackage{multirow}
\usepackage{booktabs}
\usepackage{color}
\usepackage{multicol}
\usepackage{stfloats}
\usepackage{diagbox}
\usepackage{esint}
\usepackage{float}
\usepackage{algpseudocode}
\usepackage{empheq}
\usepackage{tabularx}
\allowdisplaybreaks[4]

\def\BibTeX{{\rm B\kern-.05em{\sc i\kern-.025em b}\kern-.08em
		T\kern-.1667em\lower.7ex\hbox{E}\kern-.125emX}}
	
\ifCLASSINFOpdf \else \fi
\graphicspath{{./images/}}

\newcommand{\xt}{\mathbf{x}}

\newtheorem{theorem}{Theorem}

\newtheorem{lemma}{Lemma}
\newtheorem{proposition}{Proposition}
\newtheorem{corollary}{Corollary}
\newtheorem{definition}{Definition}

\newtheorem{assumption}{Assumption}

\newcommand*{\red}{\color{red}}
\newcommand*{\blue}{\color{black}}

\definecolor{mygreen}{RGB}{32,178,170}  % light sea green
\definecolor{mygolden}{RGB}{255,140,0} %golden 

\begin{document}
\title{OFDMA-F$^2$L: Federated Learning With Flexible Aggregation Over an OFDMA Air Interface}
\author{Shuyan Hu,~\IEEEmembership{Member, IEEE},
Xin Yuan,~\IEEEmembership{Member, IEEE}, Wei Ni,~\IEEEmembership{Fellow, IEEE}, \\
Xin Wang,~\IEEEmembership{Fellow, IEEE}, Ekram Hossain,~\IEEEmembership{Fellow, IEEE},
and H. Vincent Poor,~\IEEEmembership{Life Fellow, IEEE}
	\thanks{Copyright (c) 2015 IEEE. Personal use of this material is permitted. However, permission to use this material for any other purposes must be obtained from the IEEE by sending a request to pubs-permissions@ieee.org.

Manuscript received May 10, 2023; revised September 29, 2023; accepted November 14, 2023.
Work in this paper was supported by the National Natural Science Foundation of China under Grants No. 62231010,
No. 62071126 and No. 62101135,
the Innovation Program of Shanghai Municipal Science and Technology Commission under Grant
No. 21XD1400300,
and the U.S National Science Foundation under Grants No. CNS-2128448 and ECCS-2335876. 
Shuyan Hu and Xin Yuan contributed equally to this work.
(Corresponding author: Xin Wang).

S. Hu and X. Wang are with the Key Lab of EMW Information (MoE), the Department of Communication Science and Engineering, Fudan University, Shanghai 200433, China (e-mails: \{syhu14, xwang11\}@fudan.edu.cn).

X. Yuan and W. Ni are with the Data61, Commonwealth Scientific and Industrial Research Organization, Sydney, Marsfield, NSW 2122, Australia
(e-mails: \{xin.yuan, wei.ni\}@data61.csiro.au).

E. Hossain is with the Department of Electrical and Computer Engineering, University of Manitoba, Canada (e-mail: ekram.hossain@umanitoba.ca).

H. V. Poor is with the Department of Electrical and Computer Engineering, Princeton University, Princeton, NJ 08544 USA (e-mail: poor@princeton.edu).

  %       X. Yuan and W. Ni are with CSIRO, Sydney, Australia.\par
		% S. Hu and X. Wang are with Fudan University, Shanghai, China.\par
	 %    E. Hossain is with University of Manitoba, Winnipeg, Canada.\par
  %       H. V. Poor is with Princeton University, NJ, USA.
}}

\maketitle
%\maketitle \thispagestyle{empty} \pagestyle{empty}

\begin{abstract}

Federated learning (FL) can suffer from a communication bottleneck when deployed in mobile networks, limiting participating clients and deterring FL convergence. The impact of practical air interfaces with discrete modulations on FL has not previously been studied in depth.
This paper proposes a new paradigm of flexible aggregation-based FL (F$^2$L) over orthogonal frequency division multiple-access (OFDMA) air interface, termed as ``OFDMA-F$^2$L'', allowing selected clients to train local models for various numbers of iterations before uploading the models in each aggregation round. We optimize the selections of clients, subchannels and modulations, adapting to channel conditions and computing powers.
Specifically, we derive an upper bound on the optimality gap of OFDMA-F$^2$L capturing the impact of the selections, and show that the upper bound is minimized by maximizing the weighted sum rate of the clients
per aggregation round.
A Lagrange-dual based method is developed to solve this challenging mixed integer program of weighted sum rate maximization, revealing that a ``winner-takes-all'' policy provides the almost surely optimal client, subchannel, and modulation selections.
Experiments on multilayer perceptrons and convolutional neural networks show that 
OFDMA-F$^2$L with optimal selections can significantly improve the training convergence and accuracy, e.g., by about 18\% and 5\%, compared to potential alternatives.

\end{abstract}


\IEEEpeerreviewmaketitle

\section{Introduction}\label{sec-intro}

As a new framework for distributed online computing and model training,
federated learning (FL) has shown significant potential for applications, e.g., Internet-of-Things,
autonomous driving, remote medical care, etc.~\cite{lim20}.
% Aimed for privacy preservation and assisted by a central server,
FL enables individual mobile clients to train a global model collectively 
without releasing their data~\cite{hu21dml}.
In particular, each client trains its local model independently relying on its local dataset and sends the gradient of the local model to a server.
The server aggregates the gradients and broadcasts the aggregated global parameter to assist the clients in their local training. 

{\blue 
A challenge faced by FL is the communication bottleneck from the clients to the server,
especially when dynamic, lossy, and resource-constrained channels are considered in wireless FL~(WFL)~\cite{quzhihao22}.
The channels of geographically dispersed clients can differ significantly and change over time, leading to different model uploading delays among the clients.
Moreover, the clients can have different computation capabilities, resulting in different local training times,
and consequently, different time budgets for model uploading in an FL aggregation round.
While more participating clients and hence larger overall training datasets are conducive to better training performance~\cite{hu21dml},
this could lead to a shortage of wireless resources (e.g., bandwidth) for uploading local models.
It is crucial to jointly design the local model training and uploading (transmission) schedule of the clients to balance communication and computation.

The transmission schedule of the clients depends on the model aggregation mechanism adopted.
Typical synchronous FL (Sync-FL) requires all participating clients to upload their local models (or gradients) for a global aggregation after completing their local training~\cite{vandinh21}.
Alternatively, asynchronous FL (Async-FL) allows each client to upload its local models straightaway after completing its local training. In the latter case, there can be non-negligible gaps between the time the local models are uploaded and the time the global model is broadcast.
The computing powers of the clients may not be efficiently exploited.
Clients with more powerful computation capability have to wait longer for global model update,
especially when there are ``stragglers''~\cite{lagc20}.
Flexible aggregation-based FL (F$^2$L) has recently emerged as a promising solution to the issue of ``stragglers'', where clients can train different numbers of iterations before sending their local models synchronously
for a global aggregation~\cite{ruan2021towards}.

The transmission schedule of the clients also depends on their instantaneous computing powers~\cite{xie2020}.
A client with a weak computing capability has to transmit at a high data rate
to complete its local model training and uploading within the same period.
Consider a widely-adopted, practical orthogonal-frequency division multiple-access (OFDMA)
as the air interface for WFL.
Effective selections of the clients, channels, and modulations are key to the transmission schedule.
However, the selections form an integer program with discrete variables.
This renders the joint optimization of the communication and computation of OFDMA-based WFL systems
(for example, OFDMA-based F$^2$L, referred to as ``OFDMA-F$^2$L'' in this paper)
a typically intractable mixed-integer programming problem.

% In this sense, it is critical to holistically optimize the communication and computation of OFDMA-based F$^2$L,
% dubbed as ``OFDMA-F$^2$L'' in this paper, to speed up training and improve model accuracy.
% However, as the key of the transmission schedule, the selections of the client, channels, and modulations,
% e.g., in a widely-adopted, practical orthogonal-frequency division multiple-access (OFDMA) system,
% give an integer program with discrete variables, making the optimization of the communication and computation
% schedule a typically mathematically intractable mixed-integer programming problem.
% The loss function measuring the FL performance is nonconvex,
% imposing one more burden on the optimization task.
}

\subsection{Related Work}

Some existing work has focused on the training accuracy, efficiency, and robustness of WFL
under a persistent policy of client selection and bandwidth allocation (e.g., see \cite{hu21dml, trainyang20, qunsong20, jiaoyutao21, xutwc21, linxi23}).
A partial participating scheme of WFL clients was investigated in~\cite{vandinh21},
where only some clients were chosen to send their local models for global aggregations to minimize the power and time usage of the learning task.
A partial synchronization parallel approach was proposed in~\cite{quzhihao22} for a relay-assisted WFL, 
where local models were uploaded simultaneously and aggregated at relay nodes to reduce traffic.
The same convergence speed was achieved as the successive synchronization methods.
Yet, none of these studies \cite{trainyang20, qunsong20, jiaoyutao21, xutwc21, linxi23, vandinh21, quzhihao22}
optimized communication performances for WFL.

Considering the communication performance metrics, energy-efficient WFL has been developed to cope with lossy and dynamically varying wireless environments
in~\cite{amr20, chenhao22, qinzhijin22, zhaohui21, heyejun23, wunoma22, salehi21}.
% In particular, joint client scheduling and resource allocation for WFL were studied in~\cite{chenhao22}
% to maximize the communication efficiency by reusing stale local models.
A problem of joint training and communication was constructed in~\cite{zhaohui21}
to minimize the overall power usage of a WFL system under a training latency requirement.
% An iterative algorithm was proposed where, at every step, closed-form solutions for time, power, and bandwidth allocation, computation frequency, and learning accuracy were derived.
The optimal resource allocation policy was obtained by using bisection search.
In~\cite{heyejun23}, a hierarchical FL architecture was proposed for a heterogeneous cellular system, where the clients were jointly served by macro and micro base stations (BSs).
Deep reinforcement learning (DRL) was employed to minimize the power usage under a delay constraint
by adjusting client-BS association and resource allocation.
Non-orthogonal multiple access (NOMA)-aided FL was examined in~\cite{wunoma22},
where a set of clients formed a NOMA cluster to upload their local models to a BS
for model aggregation.
The total power usage and FL convergence latency were minimized
by jointly optimizing the NOMA transmission, BS broadcast, and FL training accuracy.
Resource-constrained WFL with partial aggregation was studied in~\cite{salehi21},
where stochastic geometry tools were employed to approximately calculate the success probability for each device.
All these studies~\cite{amr20, chenhao22, qinzhijin22, zhaohui21, heyejun23, wunoma22, salehi21}
were based on the Sync-FL framework.

The Async-FL framework was proposed to handle clients with different communication and computing capabilities~\cite{xie2020, chenyang20, chai21, gubin22, zhangyu23, hu2021source, yuan2023}.
In~\cite{wangzhyu22}, Async-FL was developed, which adapted to the heterogeneity of clients regarding their delays in model training and transmitting.
The fusion weights of the global aggregation were determined to prevent biased convergence
by requiring the weight of each client's local dataset to be in proportion to its sample number.
In~\cite{jianchun23}, an adaptive Async-FL mechanism was developed, where only a small number of the total local updates were collected by the FL server, depending on their arriving order per round.  
To reduce the mission completion time, DRL was applied to specify the number of aggregated local updates under resource constraints.

An alternative to Async-FL is F$^2$L~\cite{ruan2021towards},
where clients may have trained for different numbers of iterations when uploading their local models.
% all users upload their local models synchronously, and some of the models are trained incompletely with fewer iterations than others.
Different computing powers of the clients can be efficiently leveraged, 
although persistent global aggregation cycles are maintained as in Sync-FL.
However, no consideration has been given to the delays resulting from model uploading under flexible aggregation, which would affect the local training time and performance.
Table~\ref{tab.review} summarizes the existing works and highlights the contribution of this work.

{\blue In a different yet relevant context, over-the-air (OTA) FL has been investigated to alleviate the communication bottleneck of
WFL~\cite{zhugx, zhugxu, dujun, zhongcx}.
For instance, a broadband analog aggregation and transmission scheme was proposed in~\cite{zhugx} and~\cite{zhugxu},
where local model gradients were modulated in the analog domain and aggregated by exploiting the waveform-superposition property of radio.
An adaptive device scheduling mechanism was developed for OTA-FL in~\cite{dujun},
where data quality and energy consumption served as the criteria to select clients for OTA model aggregations
to improve the convergence and accuracy of OTA-FL.
A new model aggregation method was developed in~\cite{zhongcx} to align the local model gradients of different devices to resolve the straggler problem in OTA-FL.
Moreover, in~\cite{xichen}, a decentralized power control policy was designed to enhance the convergence of
OTA-FL under fading channels, where stochastic optimization was leveraged to decouple the power control among clients.
}

\begin{table*}[t]
\renewcommand{\arraystretch}{1.2}
\centering
\caption{Related existing works on WFL and the key contributions of this paper}\label{tab.review}
\begin{tabular}{ | c | c |  p{4.3cm}<{\centering} | p{5.8cm}<{\centering}|} 
\hline
Related work  &Year of publication  &Main topic    &Difference and improvement of this paper \\ \hline \hline
\cite{vandinh21, quzhihao22}  &2021--2022    &Sync-FL with partial synchronization or aggregation      
&F$^2$L over an OFDMA air interface
\\ \hline
\cite{ruan2021towards} &2021  &F$^2$L  &WFL over an OFDMA air interface \\ \hline
\cite{trainyang20, qunsong20, jiaoyutao21, xutwc21, linxi23}   &2020--2023   
&Sync-FL with persistent client and bandwidth assignments  
& Selections of client, channel and modulation for OFDMA-F$^2$L  \\ \hline
\cite{amr20, chenhao22, qinzhijin22, zhaohui21, heyejun23}   &2020--2023    
&Joint optimization of communication and computation for Sync-FL     
&Selections of client, channel and modulation for OFDMA-F$^2$L   \\ \hline
\cite{wunoma22}  &2022  &NOMA-aided Sync-FL  
&F$^2$L over an OFDMA air interface \\ \hline
\cite{salehi21}  &2021  & Wireless resource-constrained Sync-FL with partial aggregation
&Selections of client, channel and modulation for OFDMA-F$^2$L  \\ \hline
\cite{xie2020, chenyang20, chai21, gubin22, zhangyu23, hu2021source, yuan2023, wangzhyu22}     &2020--2023  
&Async-FL    &F$^2$L over an OFDMA air interface \\ \hline
\cite{jianchun23} &2023  &Adaptive Async-FL with partial global aggregation 
&F$^2$L over an OFDMA air interface \\ \hline

\end{tabular}
\end{table*}

\subsection{Contribution and Organization}

This paper proposes the novel OFDMA-F$^2$L approach, 
% which incorporates OFDMA into a WFL system to create an OFDMA-F$^2$L system. The system supports flexible aggregation to permit users to run different numbers of iterations in each global aggregation round, adapting to the wireless environment and computing powers of the users.
where clients can train their local models using different numbers of local iterations before uploading their local models or model gradients to an FL server for a global aggregation via an OFDMA air interface.
% To this end, the communication bottleneck of WFL systems can be significantly eased, and the computing powers of the users can be greatly harnessed.
A new algorithm is proposed to facilitate the convergence of OFDMA-F$^2$L (i.e., to minimize the convergence upper bound), by jointly optimizing the selections of clients, subchannels, and discrete modulations for local training and model uploading, adapting to the channel conditions and computing powers of the clients. The optimal number of local training iterations is accordingly specified per selected client per aggregation round. 

The key contributions of this paper can be summarized as follows:
\begin{itemize}
\item The OFDMA-F$^2$L approach is proposed,
where selected clients train local models, adapting to their available computing powers and times. The clients can also upload their local models concurrently in different subchannels,
increasing participating clients and extending their training times. 
	
\item 
We analyze a convergence upper bound on OFDMA-F$^2$L that quantifies the impact of the selections of clients, subchannels, and modulations on the convergence.

\item 
A new problem is formulated to minimize the convergence upper bound of OFDMA-F$^2$L, and transformed to maximize the weighted sum rate of the clients by optimizing the selections of clients, subchannels, and modulations.

\item By designing a Lagrange-dual based method, we 
discover that a ``winner-takes-all'' policy ensures an optimal selections of clients, subchannels, and modulations almost surely in order to achieve the fast convergence of OFDMA-F$^2$L. 
\end{itemize}

We experimentally evaluate the proposed OFDMA-F$^2$L with the optimal client, subchannel, and modulation selections on multilayer perceptron (MLP)
and convolutional neural network (CNN) models and the MNIST, CIFAR10, and Fashion MNIST (FMNIST) datasets.
It is shown that OFDMA-F$^2$L with the optimal selections significantly outperforms its alternatives with partially optimal selections or under Sync-FL, e.g., by about 18\% and 5\%, in convergence and training accuracy on the MLP model.

The remainder of this paper is organized as follows.
Section~\ref{sec-model} sets forth the system model. 
Section~\ref{sec-problem} formulates the problem of interest.
Section~\ref{sec-convergence analysis} analyzes the convergence bound of the proposed algorithm.
Section~\ref{sec-optimization} obtains the optimal assignments of clients, subchannels,
and modulations.
The OFDMA-F$^2$L system is numerically evaluated in Section~\ref{sec-sim},
followed by conclusions in Section~\ref{sec-con}.
The notations used in this paper is collected in Table~\ref{table_notations}.

\begin{table}[t]%\small
	\caption{List of notations}
	\renewcommand\tabcolsep{1.2pt}
	\begin{center}
		\begin{tabularx}{8cm}{lX}
			%\hline
			\toprule[1.5pt]
			Notation & Definition \\			
      \hline
  $\bm{\omega}$, ${\bm{\omega}}^*$ & Global model parameter and optimal global model parameter of OFDMA-F$^2$L, respectively \\
			${\cal D}_m$,~${\cal D}$ & Local dataset at client $m$, and the collection of all local datasets, respectively \\
   $F(\bm{\omega})$ & Global loss function of OFDMA-F$^2$L\\
   $\nabla F(\cdot)$ & Gradient of a function $F(\cdot)$\\				
	$t,~T$ & Index and total iteration numbers of OFDMA-F$^2$L \\
            $A$ & Maximum number of local iterations between two consecutive aggregations \\
			$\tau,~G$ & Index and maximum number of aggregations\\
			$M,~\cal M$ & Number and set of clients \\
			$K,~{\cal K}$ & Number and set of subchannels\\
			$L,~{\cal L}$ & Number and set of modulations\\
   $h_{m,k}$ & The coefficient of subchannel from client $m$ to the BS \\
			${n}_{m,k}$ & CSCG noise with zero mean and variance $\sigma^2$\\
			$\lambda_{m,k,l}$ & Binary selection indicator indicating the selection of client $m$ and modulation $l$ in subchannel $k$\\
   $\zeta_m$, $\bm{\zeta}$& Selection indicator of Client $m$ and the set of selection indicators for all clients, respectively\\
			$r_l$ & Data rate of modulation $l$\\
			$p_{m,k,l}$ & Transmit power of client $m$ in subchannel $k$ when modulation $l$ is used \\
			$P_m^{\max}$ & Power budget of client $m$ \\
			% $P_J$ & Transmit power of the jammer \\
			$\chi_{m,k,l}$ & BER of the signals uploaded by client $m$ using subchannel $k$ and 
   modulation $l$\\
			$\chi_0$ & Required BER \\           
			%\hline
			\toprule[1.5pt]
		\end{tabularx}
	\end{center}
	\label{table_notations}
\end{table}

\section{System Model and Assumptions}\label{sec-model}
\begin{figure}[tb]
	\centering
	\includegraphics[width=0.9\columnwidth]{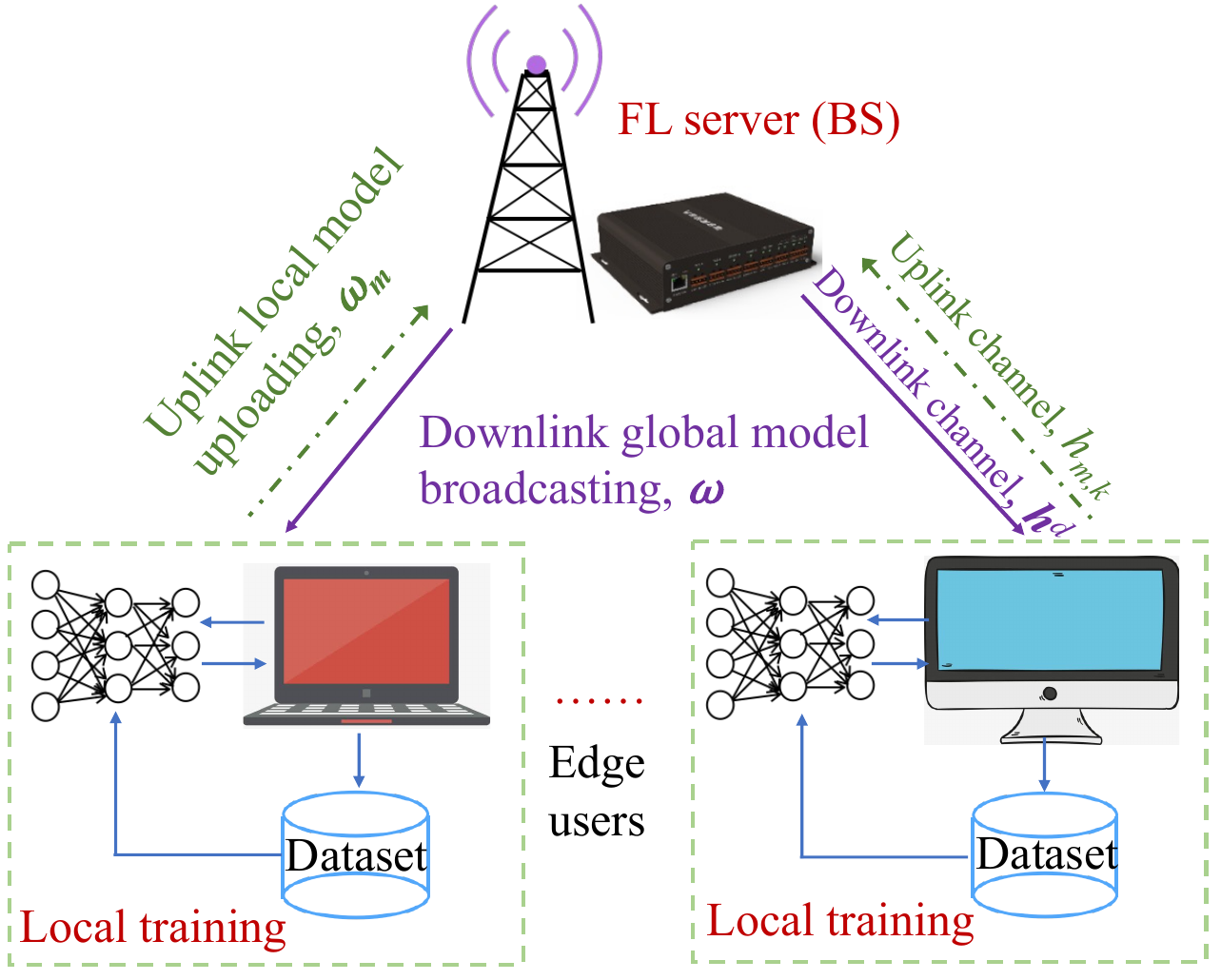}
	\caption{The proposed OFDMA-F$^2$L, where a BS serves as the FL server and multiple clients act as clients. The selected clients send their local models to the BS using OFDMA.}
	\label{fig-sysmodel}
\end{figure}

We consider flexible aggregation-based FL in an OFDMA system, i.e., OFDMA-F$^2$L, as shown in Fig.~\ref{fig-sysmodel},
where a BS serving as a central server connects $M$ edge clients over $K$ orthogonal subchannels.
The BS and clients all have a single antenna.
We use ${\cal M} = \{1,\cdots,M\}$ and ${\cal K} = \{1,\cdots,K\}$ to represent the sets of clients and subchannels, respectively.
The BS wishes to learn a machine learning (ML) model using the datasets stored at the clients, since the clients are reluctant to share their data for privacy concerns. 
As a result, the clients train and transmit their local models to the BS via the uplink OFDMA subchannels, one client per subchannel.

Upon the receipt of the local models, the BS aggregates them into a global model and distributes the global model to the clients through a downlink broadcast. 
% by aggregating the local models trained at the edge users, and then broadcasting the global model to all the edge users. The edge users upload their local models to the FL server via wireless links using an orthogonal frequency division multiple access (OFDMA) technique in the uplinks, where each edge user occupies only one subchannel and broadcasts in the downlink. 
Considering limited communication and computation resources, only some clients are chosen to train and upload their local models. The BS performs client selection and resource allocation.

\subsection{F$^2$L Model}
Let ${\cal D}_m$ denote the local dataset at client $m \in {\cal M}$, and $D_m = |{\cal D}_m|$ denote the size of the local dataset, where $|\cdot|$ stands for cardinality. ${\cal D} = \underset{\forall m \in  {\cal M}}{\cup} {\cal D}_m $ collects all local datasets at the clients with $D = |{\cal D}| = \sum_{m \in {\cal M} }D_m$ being the size of ${\cal D}$.
% Each edge user has a certain computation capability and could train a neural network model locally. 
Assume that each client $m$ divides its local dataset ${\cal D}_m$ into multiple independent and identically distributed (i.i.d.) mini-batches with the same distribution as ${\cal D}_m$. 
$F_m(\bm{\omega};{\cal X})$ is the loss function of client $m \in {\cal M}$, measuring the model error on the client's dataset ${\cal X} \subseteq {\cal D}_m$.
Here, $\bm{\omega}$ is the global model parameter vector. 
The global loss function at the BS, denoted by $F(\bm \omega)$, is 
% associated with all local datasets residing at the edge clients, and can be
acquired by gathering the local models uploaded from the selected clients, as given by
% \begin{equation}\label{eq-problem}
% {\bm \omega}^* \triangleq \arg \underset{\bm \omega, \,\bm{\zeta^{\tau}}}{\min}~ F({\bm \omega}) 
% ={\blue \arg \underset{\bm \omega, \,\bm{\zeta^{\tau}}}{\min} \sum_{m \in {\cal M'}={\cal M} (\bm{\zeta^{\tau}})} \varrho_{m}  F_{m}({\bm \omega})
% = \arg \underset{\bm \omega,\,\bm{\zeta^{\tau}}}{\min} \sum_{m \in {\cal M}} \varrho_{m} {\zeta_m^{\tau}} F_{m}({\bm \omega}),}
% \end{equation}
\begin{equation}\label{eq-problem}
\begin{aligned}
& {\bm \omega}^* = \arg \underset{\bm \omega,\,\{\zeta^{\tau},\forall \tau\}}{\min}~ F({\bm \omega}),\\
\text{ with }
& F({\bm \omega}) \triangleq 
\frac{1}{B}\sum_{m \in {\cal M}} \sum_{{\tau}} \varrho^{\tau}_{m} \zeta^{\tau}_m 
F_{m}({\bm \omega};{\cal B}^{\tau}_m),
\end{aligned}
\end{equation}
where $\bm{\omega}^*$ is the optimal model parameter minimizing $F(\bm \omega)$; 
% $\varrho^{\tau}_{m}$ is the aggregation coefficient for client $m$ at the $\tau$-th aggregation round; 
${\boldsymbol{{\cal B}}}^{\tau}_m \subseteq {\cal D}_m$ collects the set of mini-batches sampled from the local dataset ${\cal D}_m$ during the $\tau$-th aggregation round;
$B = \sum_m\sum_{\tau} I^{\tau}_m$ and $I^{\tau}_m $ is the number of local training iterations at client $m$ in the $\tau$-th aggregation round.
% the number of mini-batches used for client $m$'s local training during the $\tau$-th aggregation round.
$\varrho^{\tau}_m = \frac{A}{I^{\tau}_m} \cdot \frac{D_m}{D}$ is the aggregation coefficient of client $m$ for the $\tau$-th aggregation of F$^2$L, where  $A$ is the maximum number of iterations a client can train locally per aggregation round and $I^{\tau}_m \leq A$.
In this sense, greater aggregation coefficients can be given to clients that train fewer iterations in an aggregation round, compensating for the additional iterations other clients completed and contributing to unbiased gradient after aggregation~\cite[Sec. 4.1, Scheme C]{ruan2021towards}.
% where $A$ is the maximum number of iterations a client can train locally per aggregation round and $I^{\tau}_m \leq A$ is the number of local training iterations at client $m$ in the $\tau$-th aggregation round.
% {\red This adaptive aggregation coefficients contributes to unbiased gradient after aggregation~\cite{ruan2021towards}. 
% This is because increasing $\varrho^{\tau}_m$ is equivalent to increasing the learning rate of client $m$, as will be shown later in \eqref{eq_global_weight}.

Here, 
$\bm{\zeta^{\tau}} = \{\zeta^{\tau}_1,\cdots, \zeta^{\tau}_M\}$ collects the client selection indicators with
$\zeta^{\tau}_m$ indicating client $m$ is selected for the $\tau$-th aggregation; i.e.,  $\zeta^{\tau}_m = 1$ if client $m$ is selected; otherwise, $\zeta^{\tau}_m = 0$.
% ; and
% ${\cal M'}={\cal M} (\bm{\zeta^{\tau}})$ is the set of selected clients. 
The client selection is critical for FL, 
especially when the communication and computation resources are constrained. Not all clients can upload their local models to the BS.

{\blue OFDMA-F$^2$L supports flexible aggregation at the BS~\cite{ruan2021towards},
where the clients can execute different numbers of local iterations per aggregation round 
and, if selected, transmit their latest local models to the BS for a synchronous global model aggregation
(in other words, the clients complete uploading their local models at the same time
so that the global aggregation can start straightaway).
The FL process is divided evenly into $G$ global aggregations.
Each aggregation round lasts for the same duration.}
The OFDMA-F$^2$L contains three alternating steps: 
\begin{itemize}
\item \textit{Local Model Training: }After the $\tau$-th global model aggregation
($\tau = 1, \cdots, G$),
each selected client $m$ ($\zeta^{\tau}_m = 1$) performs $I^{\tau}_m$ iterations of local training.
The local training of each selected client $m$ for the
$(\tau +1)$-th global aggregation is initialized to be $\bm{\omega}^{\tau}_m(0) = \bm{\omega}^{\tau}$, with $\bm{\omega}^{\tau}$ being the global model acquired in
the $\tau$-th global model aggregation. The local model of client $m$ at the $i$-th iteration
($i=1, \cdots, I^{\tau}_m$), denoted by $\bm{\omega}^{\tau}_m(i) $, is refreshed by
% \begin{equation}\label{eq_local_SGD}
% \bm{\omega}^{t}_m = {\bm{\omega}}^{t-1}_m - \eta \nabla F_m ({\bm{\omega}}^{t-1}_m),
% \end{equation} 
\begin{equation}\label{eq_local_SGD}
\bm{\omega}^{\tau}_m(i) = {\bm{\omega}}^{\tau}_m (i - 1)  - \eta \nabla F_m \left({\bm{\omega}}^{\tau}_m (i - 1);{\cal B}^{\tau}_m (i-1) \right),
\end{equation} 
where $\eta$ is the learning rate, and ${\boldsymbol{{\cal B}}}^{\tau}_m = \{{{\cal{ B}}}^{\tau}_m (0), \cdots, {\cal B}^{\tau - 1}_m (I^{\tau}_m -1)\}$ collects all mini-batches sampled from the local dataset ${\cal D}_m$.
${\cal B}^{\tau}_m (i)$ is the mini-batch used by device $m$ for the $i$-th local iteration of the $\tau$-th global aggregation.
% The $m$-th user clips the local model parameter ${\bm \omega}^t_m$ with a pre-determined threshold $C$, i.e., $\left\|\bm{\omega}^{t+1}_m \right\| \leq C$.
    
    \item \textit{Global Aggregation:}
The clients selected by the BS upload their local models for the $(\tau+1)$-th global model aggregation after their local training. The global model obtained is
\begin{equation}\label{eq_global_weight}
\begin{aligned}
	\bm{\omega}^{\tau+1 } & = \sum_{m \in {\cal M}} \varrho^{\tau}_m \zeta^{\tau}_m 
   {\bm{\omega}^{\tau}_m } \\
   & = \sum_{m \in {\cal M}} \varrho^{\tau}_m \zeta^{\tau}_m \left(\sum^{I^{\tau}_m -1}_{i=0} \left( {\bm{\omega}^{{\tau }}_m {(i)} } - \eta \nabla F_m ({\bm{\omega}}^{\tau }_m{(i)}; {\cal B}^{\tau}_m{(i)}) \right)\right).
\end{aligned}
\end{equation}

\item \textit{Global Model Update: }The BS broadcasts the global model $\bm{\omega}^{\tau}$ and selects the clients for the next global aggregation.
Then, the selected clients start their local training iterations based on~$\bm{\omega}^{\tau}$.    
\end{itemize}
\textbf{Algorithm~\ref{algo_DGD}} summarizes the OFDMA-F$^2$L method.
As will be described in Section~\ref{subsec - definition and assumption}, we assume a convex loss function of the ML task, and hence,  OFDMA-F$^2$L converges. Suppose that the optimal ML model parameter $\bm{\omega}^*$ can be acquired after $G$ global aggregations. $F(\bm \omega^*)$ provides the minimum global loss. Nevertheless, when the loss function is non-convex, OFDMA-F$^2$L can still converge, as typically observed in the literature, e.g., \cite{chenyang20,yuan2023}, and numerically validated in Section~\ref{subsec-cnn}.
Note that if $I^{\tau}_m = A, \forall m \in {\cal M}$, i.e., all selected clients finish a given number of $A$ local iterations (or updates), Algorithm~\ref{algo_DGD} becomes the FedAvg algorithm~\cite{mcmahan2017communication}.

%%%%%%%%%%%%%%%%%%%%%%%%%%%%%%%%%%%%%
\begin{algorithm}[!t]\small
	\caption{OFDMA-F$^2$L.}
	\label{algo_DGD}
	\KwIn{$I^{\tau}_m$, $G$, $\bm{\omega}^0$, $\epsilon$, $\delta$, and $C$.}
	\KwOut{$\bm{\omega}^T$.}
 	\LinesNumbered
	Initialize: $t=0$, $\bm{\omega}^0$\;
        \For{$\tau = 1:G$}{
	\While {$i \leq I^{\tau}_m$}{
		\tcp{ Local model update}
		            \quad $\bm{\omega}^{\tau}_m(i) = {\bm{\omega}}^{\tau}_m (i - 1)  - \eta \nabla F_m ({\bm{\omega}}^{\tau}_m (i - 1) )$\;
		 Clip the local parameters:
		             \qquad\quad $\bm{\omega}^{\tau}_m(i) = \bm{\omega}^{\tau}_m(i) / \max\left(1,\frac{\left\|\bm{\omega}^{\tau}_m(i) \right\| }{C} \right)$\;
	$i \leftarrow i+1$;
}
\tcp{ Global model aggregation} 
 
{	Update the global model parameters: %\qquad
		  $\bm{\omega}^{\tau } = \sum_{m \in {\cal M}} \varrho^{\tau}_m \zeta^{\tau}_m \bigg(\sum^{I^{\tau}_m -1}_{i=0} \Big( {\bm{\omega}^{{\tau }}_m {(i)} } - \eta \nabla F_m ({\bm{\omega}}^{\tau }_m{(i)}; {\cal B}^{\tau}_m{(i)}) \Big)\bigg)$\;
    Update the local model parameters:
	\qquad  ${\bm{\omega}}^{\tau}_m(0) = \bm{\omega}^{\tau}, \; \forall m \in {\cal M}$\;}
}
\end{algorithm}
%%%%%%%%%%%%%%%%%%%%%%%%%%%%%%%%%%%%

% \begin{algorithm}[!t]\small
% 	\caption{OFDMA-F$^2$L.}
% 	\label{algo_DGD}
% 	\LinesNumbered
% 	\KwIn{$I^{\tau}_m$, $G$, $\bm{\omega}^0$, $\epsilon$, $\delta$, and $C$.}
% 	\KwOut{$\bm{\omega}^T$.}
% 	%some description\; 
% %	\hspace*{0.02in} {\bf Procedure:}
% 	Initialize: $t=0$, $\bm{\omega}^0$\;%, and $\hat{\bm{\omega}}(0)$\;
%         \For{$\tau = 1:G$}{
% 	\While {$i \leq I^{\tau}_m$}{
% % 		\While {$\tau < G $}{
% 		 \tcp{ Local model update}
% 		            \quad $\bm{\omega}^{\tau}_m(i) = {\bm{\omega}}^{\tau}_m (i - 1)  - \eta \nabla F_m ({\bm{\omega}}^{\tau}_m (i - 1) )$\;
% 		 Clip the local parameters:
% 		            \qquad\quad $\bm{\omega}^{\tau}_m(i) = \bm{\omega}^{\tau}_m(i) / \max\left(1,\frac{\left\|\bm{\omega}^{\tau}_m(i) \right\| }{C} \right)$\;
% % 		}
         
% 	   %   \% Global aggregation\;
% 	   %   $\tau = \tau +1$\;}
% % 		{${\bm{\omega}}^{t+1}_m = {\bm{\omega}}^{t+1}_m, \; \forall m \in {\cal M}$\;}
% 	$i \leftarrow i+1$\;
% }
% \tcp{ Global model aggregation} 
% 	% \If {$(t + 1)$ is an integer multiple of $A$}
%  {
% 	Update the global model parameters: %\qquad
% 		  $\bm{\omega}^{\tau } = \sum_{m \in {\cal M}} \varrho^{\tau}_m \zeta^{\tau}_m \Big(\sum^{I^{\tau}_m -1}_{i=0} \left( {\bm{\omega}^{{\tau }}_m {(i)} } - \eta \nabla F_m ({\bm{\omega}}^{\tau }_m{(i)}; {\cal B}^{\tau}_m{(i)}) \right)\Big)$\;
%     % \NoNumber{\% Global model aggregation\;} 
%     Update the local model parameters:
% 	\qquad  ${\bm{\omega}}^{\tau}_m(0) = \bm{\omega}^{\tau}, \; \forall m \in {\cal M}$;}
 
% }
% \end{algorithm}

\subsection{Transmission Model}

At the beginning of an aggregation round, the BS selects the clients, subchannels, and modulation modes based on the channel state information (CSI) and computing powers that the clients fed back at the end of the previous aggregation round.
Then, the BS broadcasts the selections and the global model aggregated at the end of the previous aggregation round, e.g., using its lowest non-zero modulation mode and full transmit power.
As a result, all clients are expected to reliably receive the global model and the selections/schedule. 
% At the start of each aggregation, the BS selects the clients, subchannels, and modulation modes based on the channel state information (CSI)
% that the clients sent back towards the end of the previous aggregation round.
% The BS broadcasts its selection and the global model aggregated at the end of the previous aggregation round to all clients,
% e.g., using its lowest non-zero modulation mode and full transmit power.
The downlink delay, denoted by $T^{\text{DL}}$, is given.
The global and local models have the same size, that is, $|\bm{\omega}| = |\bm{\omega}_m|$.

Towards the end of each aggregation round, the selected clients upload their local models using an OFDMA protocol. Let ${\xt}_k^{\tau} \triangleq \left[x_{1,k}^{\tau},\cdots,x_{M,k}^{\tau}\right]^T \in {\mathbb{C}^{M \times 1}}$
represent the signals that each of the $M$ clients would send in subchannel $k$,
if it is allocated with the subchannel at the $\tau$-th aggregation round.
Consider Rayleigh fading, the channel gain of client $m$ in subchannel $k$ is
\begin{align}\label{eq-bs-ue}
h_{m,k}^{\tau} = \sqrt{\epsilon_o \left(d_{m} \right)^{-\alpha}} \tilde{h}^{\tau},~\forall m,k,
\end{align}
where $\tilde{h}^{\tau}$ is a random scattering element captured by zero-mean and unit-variance circularly symmetric complex Gaussian (CSCG) variables. $\epsilon_o$ is the path loss at the reference distance $d_0 =1$ m. $\alpha$ is the path loss exponent.
$d_{m}$ is the distance between client $m$ and the BS.
% in the three-dimensional (3D) Cartesian coordinate system.

Assume that the channels experience block fading, i.e., the channels stay the same within a global aggregation round and vary independently between aggregation rounds~\cite{Viswanathan1999Capacity}. 
At the ${\tau}$-th aggregation round, the received SNR at the BS from client $m$ in subchannel $k$ is given by
\begin{equation}\label{eq-snr}
     \gamma_{m,k}^{\tau} = \frac{ p_{m,k}^{\tau}|h_{m,k}^{\tau}|^2}{\sigma^2},
\end{equation}
where $p_{m,k}^{\tau}$ is the transmit power of client $m$ in subchannel~$k$,
and $\sigma^2$ is the variance of the zero-mean CSCG noise ${n}_{m,k}^{\tau}$, i.e.,
${n}_{m,k}^{\tau} \in {\cal CN}\left(0,\sigma^2 \right)$.

Each client can take a modulation mode from a discrete set of $L$ modulation modes. The set is denoted by ${\cal L} = \{0, 1, \cdots, L\}$, where $L=|\mathcal{L}|$ and $|\cdot|$ stands for cardinality.
The data rate of modulation $l$ is represented by $r_l$. 
$l = 0$ indicates no transmission, i.e., $r_0 =0$.
At the BS, the bit error rate (BER) of the local model from client $m$ using subchannel $k$
and modulation $l$ is~\cite{Goldsmith1998}
\begin{equation}\label{eq-ber}
	\chi_{m,k,l}^{\tau} = \beta_1 \exp\left(-\frac{\beta_2 \gamma_{m,k}^{\tau}}{2^{r_l} - 1} \right), 
\end{equation}
where $\beta_1$ and $\beta_2$ are constants relying on the modulation.

To ensure the local model parameters transmitted from client $m$ to the BS using
subchannel $k$
and modulation $l$ achieve the required BER $\chi_0$,
the minimum transmit power required is~\cite{Malik2018Interference}
\begin{equation}\label{eq-power}
p_{m,k,l}^{\tau} 
% \left(|h_{m,k}^{\tau}|^2\right) 
= \frac{\left(2^{r_l} -1 \right) \ln\left(\frac{\beta_1}{\chi_0} \right)\sigma^2 }{\beta_2 \left|h_{m,k}^{\tau} \right|^2},
\end{equation}
which is acquired by outplacing $\chi_{m,k,l}^{\tau}$ with $\chi_0$ in \eqref{eq-ber} and then substituting \eqref{eq-snr} into \eqref{eq-ber}, followed by rearranging \eqref{eq-ber}.
Here, we choose a sufficiently small value of $\chi_0$ (e.g., $\le 10^{-6}$).
In other words, the BS receives reasonably accurate and reliable local models.

Let $\lambda_{m,k,l}^{\tau} = 1$ denote the assignment of subchannel $k$ and modulation $l$ to client $m$ to upload its local model, given $|h_{m,k}^{\tau}|^2$;
and $\lambda_{m,k,l}^{\tau} = 0$ denotes otherwise.
Let $\bm{\lambda}^{\tau} := \left\lbrace \lambda_{m,k,l}^{\tau}, \;\forall m \in {\cal M},\;
k \in {\cal K},\; l \in {\cal L}\right\rbrace $
% $\bm{\lambda} := \left\lbrace \bm{\lambda}^{\tau}, \; \forall \tau \in \mathcal{G}\right\rbrace $
represent the collection of all indicators to be optimized in the $\tau$-th aggregation round.
The data rate of client $m$ is obtained as 
\begin{equation}\label{eq.datarate}
R^{\text{UL},\tau}_m = \sum_{k = 1}^{K}\sum_{l=0}^{L} \lambda_{m,k,l}^{\tau} \times r_l.
\end{equation}
According to the minimum transmit power specified in \eqref{eq-power}, the transmit power of client $m$ is 
\begin{equation}
	P_m^{\tau} = \sum_{k = 1}^K \sum_{l=0}^{L} \lambda_{m,k,l}^{\tau} \times p_{m,k,l}^{\tau}. 
\end{equation}

The delay for client $m$ to upload its local model is 
\begin{equation}\label{eq.tmurmu}
T^{\text{UL},\tau}_{m} = \frac{|\bm{\omega}_m|}{R^{\text{UL},\tau}_m} = \frac{|\bm{\omega}_m|}{\sum_{k = 1}^{K}\sum_{l=0}^{L} \lambda_{m,k,l}^{\tau} \times r_l},
\end{equation}
where $|\bm{\omega}_m|$ is the size of client $m$'s local model (in bits).

\subsection{Computing Model}

We adopt the number of floating-point operations (FLOPs) to reflect the computational capability of the clients.
$\mu$ denotes the number of FLOPs required to train a local model relying on a mini-batch of a local dataset.
In the $\tau$-th aggregation round, the number of FLOPs required for client $m$  is $\mu_m^{\tau} = \mu I_m^{\tau}$.
The local training delay of client $m$ is
\begin{equation}\label{eq.tmc}
    T^{{\rm C}, \tau}_m = \frac{\mu_m^{\tau}}{\beta_m} = \frac{\mu I_m^{\tau}}{\beta_m}.%\le \frac{\mu |{\cal B}_m^{\tau}|}{\beta_m} .
\end{equation}
Here, $\beta_m$ is the computational (or processing) speed of client~$m$.
% $I_m^\tau$ is the number of mini-batches for client $m$'s local training in the $\tau$-th aggregation round.

% Since $T_m^{\rm C}$ is assessed every aggregation round, for brevity of notation, we replace $I_m^{\color{red}\tau}$ and ${\cal B}_m^\tau$ with $I_m$ and ${\cal B}_m$, respectively, in the rest of this paper.} 

As per any global aggregation $\tau$, the total delay undergone by client $m$ is given by
\begin{equation}\label{eq.tmm}
    T_m^{\tau} = T^{{\rm C}, \tau}_m + T^{\text{UL},\tau}_m + T^{\text{DL}}, 
\end{equation}
which considers both the transmission and computation delays.
{\blue Given the equal duration of each FL round,
if the uplink and downlink delays are longer,
there is a shorter time for local model training.
If the computing delay is longer per iteration, the client can only complete fewer iterations.
This impacts the convergence and accuracy of the training.}

\section{Problem Statement}\label{sec-problem}
% In this paper, we aim to design an efficient OFDMA-based FL system, where a central FL server communicates with multiple edge users through wireless channels.
We propose to jointly design the selection of clients, subchannels and modulations,
i.e., $\bm{\lambda}^{\tau}$, 
to minimize the global loss function of the OFDMA-F$^2$L system
per aggregation round, i.e., the $\tau$-th round.
The problem is formulated as 
\begin{subequations}\label{eq-P1}
	\begin{align}
	\textbf{P1}:\;\min_{\{\bm{\lambda}^{\tau}, {\bm \omega} \}}\;\; 
&   \sum_{m \in {\cal M}} \frac{D_m \zeta_m^{\tau} F_{m}({\bm \omega})}{D} \label{eq-P1 a} \\ 
	{\text{s.t.}}\; & \zeta_m^{\tau} = {\rm sgn}\left(\sum_{k=1}^{K} \sum_{l=0}^{L} \lambda_{m,k,l}^{\tau}\right), \label{eq-P1 b1}\\
 & P_m^{\tau}  \leq P^{\max}_m,\forall m,\label{eq-P1 b}\\
	&\sum_{k=1}^{K} \sum_{m = 1}^M\sum_{l=0}^{L} \lambda_{m,k,l}^{\tau} \leq K,\label{eq-P1 d}\\
	&\sum_{m = 1}^M\sum_{l=0}^{L} \lambda_{m,k,l}^{\tau} \leq 1,~\forall k,\label{eq-P1 e}\\
	& \lambda_{m,k,l}^{\tau} \in \{0,1\}, \label{eq-P1 f}\\
% 	& R^{(1)}_m\left(\bm{\eta}\right) = \chi R^{(2)}_m\left(\bm{\eta}\right), \label{eq-P1 g}\\
	& T_m^{\tau} \leq T_{\rm th}, \label{eq-P1 h}        
	\end{align}
\end{subequations}
where ${\rm sgn} (\cdot)$ is the sign function.
{\blue The objective~\eqref{eq-P1 a} stems from~\eqref{eq-problem}.
Specifically, \eqref{eq-problem} minimizes the global loss function over all aggregation rounds. 
\eqref{eq-P1 a} decouples the minimization over time and optimizes the variables and parameters for each round,
which is in line with practical FL implementations resulting from causality.}
Hence, constraint~\eqref{eq-P1 b1} yields
\begin{subequations}\label{eq.zeta}
    \begin{numcases}{\zeta^{\tau}_m =} 
      & $1, ~\text{if}~\sum_{k=1}^{K}\sum_{l=0}^{L} \lambda^{\tau}_{m,k,l}   \geq 1; $
       \label{eq.zeta.a} \\
      & $0, ~\text{if}~\sum_{k=1}^{K}\sum_{l=0}^{L} \lambda^{\tau}_{m,k,l}    = 0.$ 
       \label{eq.zeta.b} 
    \end{numcases}
\end{subequations}
In other words, client $m$ is selected if at least one subchannel and modulation are allocated to the client.
Constraint~\eqref{eq-P1 b} indicates that the sum transmit power of each client is upper bounded by~$P^{\max}_m$.
\eqref{eq-P1 d} specifies that the number of subchannels allocated to all clients is no greater than~$K$.
Constraint \eqref{eq-P1 e} ensures that a subchannel is allocated to at most one client to avoid inter-client interference.
\eqref{eq-P1 h} indicates that the delay (including communication and computation delay) required
to conduct the F$^2$L is upper bounded by a pre-specified threshold,
i.e., $T_{\rm th}$, for each global aggregation. 

The determination of $\lambda_{m,k,l}^{\tau}$ relies on not only the transmission rate (or time), but also the computing capability of each client.
This indicates that the computing time $T_m^{\rm C, \tau}$ or the
number of iterations per selected client $I_m^{\tau}$
is also optimized implicitly by solving problem \textbf{P1} per round~$\tau$.
Once $\lambda_{m,k,l}^{\tau}$ is decided, the transmit power of client $m$ using subchannel $k$
and modulation $l$, that is, \eqref{eq-power}, is set to satisfy the BER requirement.

Problem \textbf{P1} is new, yet mathematically intractable.
It is non-convex in $\bm{\lambda}^{\tau}$ and ${\bm \omega}$,
and has a mixed integer programming nature resulting from the client selection and subchannel assignment. 
Moreover, the global loss function $F({\bm{\lambda}^{\tau}},{\bm \omega})$ is non-convex
and has a large number of parameters to be optimized.
As a consequence, problem \textbf{P1} cannot be directly solved using conventional optimization methods,
such as alternating optimization and successive convex approximation.

We decouple problem \textbf{P1} between the global model learning,
and the assignments of the client and modulation mode per subchannel.
We first derive the convergence upper bound of the global loss function of
OFDMA-F$^2$L, i.e., $F({\bm \omega})$,
and then derive the almost surely optimal assignments of clients, channels, and modulations
to minimize the upper bound.

\section{Convergence Analysis of OFDMA-F$^2$L}\label{sec-convergence analysis}

This section starts by deriving the optimality gap of OFDMA-F$^2$L.
% which can guide the selections of users, channels, and modulation-coding schemes. Specifically,
The optimality gap is minimized per communication by optimally selecting clients, channels, and modulations, to encourage the convergence of the OFDMA-F$^2$L, as will be delineated in Section~V.

\subsection{Definitions and Assumptions}\label{subsec - definition and assumption}
This section establishes the upper bound for the optimality gap of OFDMA-F$^2$L, i.e., \textbf{Algorithm~\ref{algo_DGD}}.
We start by presenting the definitions and assumptions used.

\begin{definition}[Gradient Divergence]\label{def-divergence}
	 Let $\gamma_m$ denote an upper bound of the gradient difference between the local and global loss functions, i.e.,
	$\left\|\nabla F_m({\bm \omega})-\nabla F({\bm \omega}) \right\| \leq \gamma_m$, $\forall m, \bm \omega$.
	The global gradient divergence is $\gamma \triangleq  \frac{\sum_{m} D_m \gamma_m}{{D}}$.
\end{definition}

\begin{assumption}\label{assumption}
	$\forall m \in {\cal M}$, we make the following assumptions:
	\begin{enumerate}
	    
		\item The loss function of FL is smooth. Specifically, the gradient of $F_m(\bm \omega)$ is
         $L_c$-Lipschitz continuous~\cite{o2006metric}, that is, $\left\|\nabla F_m ({\bm \omega}^{\tau+1})-\nabla F_m ({\bm \omega}^{\tau}) \right\| \leq L_c \left\| {\bm \omega}^{\tau+1} - {\bm \omega}^{\tau}\right\|,\,\forall {\bm \omega}^{\tau},{\bm \omega}^{\tau+1} $, with $L_c$ being a constant depending on the loss function, so as the gradient of $F({\bm \omega})$ is also $L_c$-Lipschitz continuous;          
         
		\item The learning rate is $\eta \leq \frac{1}{L_c}$;
		
		\item The mean squared norm of the stochastic gradients at each edge client $m$ is uniformly bounded by $\mathbb{E}_{\mathcal{B}^{\tau}_m}\left\| \nabla F_m ({\bm{\omega}}^{\tau }_m(i); \mathcal{B}^{\tau }_m(i)) \right\|^2  \leq \kappa_1 + \kappa_2 \|\nabla F(\bm{\omega}^{\tau })\|^2, \forall m~\text{and}~\tau$ with $\kappa_1,\kappa_2 \geq 0$;

		\item $F_m(\bm \omega)$ meets the Polyak-Lojasiewicz condition~\cite{karimi2016linear}
  with a constant $\rho$, implying that $F(\bm \omega)-F({\bm \omega}^*) \leq \frac{1}{2\rho} \left\| \nabla F(\bm \omega) \right\|^2 $;
  
		\item $F({\bm \omega}^0) - F({\bm \omega}^*) =  \Theta$, in which $\Theta$ is a constant;
  in other words, the initial optimality gap is bounded.
	\end{enumerate}
\end{assumption}

{\blue These assumptions have been extensively adopted in the literature,
e.g., \cite{quzhihao22, jianchun23, yuan2023}.
Many modern ML models have multiple layers and non-convex loss functions (e.g., CNN).
Nonetheless, a considerable number of ML tasks, e.g., logistic regression, have (strongly) convex loss functions, e.g., cross-entropy, ordinary least, etc.
These convex loss functions satisfy \textbf{Assumption 1}.
Many existing algorithms developed under the (strong) convexity assumptions were experimentally validated for ML models with non-convex loss functions,
e.g., \cite{jianchun23, yuan2023}, and references therein.}
% Following this convention, we test our proposed policy under a non-convex image classification task using CNN in Section VI, which validates its applicability to modern neural networks with non-convex loss functions.}

\subsection{Convergence Analysis}

Based on Definition~\ref{def-divergence} and Assumption~\ref{assumption}, each local loss function in \eqref{eq-problem} is strongly convex and hence, OFDMA-F$^2$L converges.
We establish the ensuing theorem to evaluate the upper bound for the optimality gap of OFDMA-F$^2$L,
i.e., the gap between ${\bm \omega}^{\tau}$ and ${\bm \omega}^*$. 
\begin{theorem}\label{theo_convergence bound}
	After $\tau$ global aggregations, the convergence upper bound of OFDMA-F$^2$L is 
	\begin{equation}\label{eq_convergnece_bound 0}
	\begin{aligned}
	 F({\bm \omega}^{\tau } )& - F({\bm \omega}^*) \leq (1 - 2\rho \phi_2)^{\tau} \Theta + \phi_1 \frac{1-\left(1 - 2\rho \phi_2\right)^{\tau}}{ 2\rho \phi_2},
	\end{aligned}
	\end{equation}	
	where $\phi_1 \triangleq \frac{\eta^2 L_c}{2} \sum_{m \in {\cal M}}  \frac{A^2 \kappa_1 D_m^2 \zeta^{\tau}_m}{D^2 I^{\tau}_m}$ and $\phi_2 \triangleq \eta - \frac{\eta^2 L_c}{2} \sum_{m \in {\cal M}}  \frac{A^2 \kappa_2 D_m^2 \zeta^{\tau}_m}{D^2 I^{\tau}_m}$.
\end{theorem}
\begin{proof}
	See Appendix~\ref{appendix_convergence bound}.
\end{proof}

To ensure the convergence of OFDMA-F$^2$L, the following proposition is developed based on
\textbf{Theorem~\ref{theo_convergence bound}}.
\begin{proposition}
    To guarantee the convergence of OFDMA-F$^2$L, it must hold that $1 - 2\rho \phi_2 < 1$. Hence, $\kappa_2$ must satisfy 
    \begin{equation}
        \kappa_2 < \frac{2}{\eta L_c A^2 \sum_{m \in {\cal M}} \frac{D^2_m \zeta^{\tau}_m}{D^2 I^{\tau}_m}} \overset{\eta = \frac{1}{L_c}}{=} \frac{2}{A^2 \sum_{m \in {\cal M}} \frac{D^2_m \zeta^{\tau}_m}{D^2 I^{\tau}_m}}.
    \end{equation}
\end{proposition}

This proposition provides a sufficient condition for the convergence of OFDMA-F$^2$L.

\begin{corollary}\label{rema-bound}
In the case of $1 - 2\rho \phi_2 < 1$, as $\tau \to \infty$, the convergence upper bound of OFDMA-F$^2$L is obtained as
\begin{subequations}
\begin{align}
    F({\bm \omega}^{\tau} ) - F({\bm \omega}^*) & \leq \frac{\phi_1}{2 \rho \phi_2}  \\
    & = \frac{\eta^2 L_c}{4\rho} \frac{\sum_{m \in {\cal M}}\frac{A^2\kappa_1 D_m^2 \zeta^{\tau}_m}{D^2 I^{\tau}_m}}{\eta - \frac{\eta^2 L_c}{2}\sum_{m \in {\cal M}}\frac{A^2\kappa_2 D_m^2 \zeta^{\tau}_m}{D^2 I^{\tau}_m}} \\
   & \leq \frac{1}{2 \rho} \sum_{m \in {\cal M}}\frac{D_m^2 \zeta^{\tau}_m}{D^2 I^{\tau}_m}  \propto  \sum_{m \in {\cal M}}\frac{D_m^2 \zeta^{\tau}_m}{D^2 I^{\tau}_m}.    
\end{align}    
\end{subequations}
The upper bound is proportional to $\sum_{m \in {\cal M}}\frac{D_m^2 \zeta^{\tau}_m}{D^2 I^{\tau}_m}$.
\end{corollary}

According to Corollary~\ref{rema-bound}, minimizing the optimality gap of OFDMA-F$^2$L is equivalent to
minimizing $\sum_{m \in {\cal M}}\frac{D_m^2 \zeta^{\tau}_m}{D^2 I^{\tau}_m}$ by optimally selecting clients, channels, and modulations.
Consequently, we can reformulate the problem of interest to facilitate the problem-solving in Section~\ref{sec-optimization}.

% {\blue Note that many ML models, e.g., CNNs, have non-convex loss functions. Nevertheless, many ML tasks still have convex loss functions, e.g., MLP and support vector machine (SVM). The convergence of algorithms developed under convex consumption has been typically studied due to their mathematical traceability. The algorithms have often been numerically validated for ML models with non-convex loss functions, e.g., \cite{chenyang20,yuan2023}. For this reason, we also evaluate our algorithm developed in Section~\ref{sec-sim} under a non-convex image classification task that employs a CNN, as will be described in Section~\ref{subsec-cnn}.}

\section{Optimal Client, Channel, and Modulation Selection for OFDMA-F$^2$L}\label{sec-optimization}
Since the selections of the clients, subchannels, and modulations are performed independently in every aggregation round, i.e., the $\tau$-th aggregation round, we suppress the superscript ``$^\tau$'' in 
% the channel coefficients, transmit powers, SNRs, selection indicators, and computation and communication delays 
the rest of this paper for brevity of notation.

According to \textbf{Theorem~\ref{theo_convergence bound}},
minimizing the global loss function can be transformed into minimizing the optimality gap of OFDMA-F$^2$L,
i.e., between
$F({\bm \omega} )$ and $F({\bm \omega}^*)$.
Since the optimality gap is proportional to $\sum_{m \in {\cal M}}\frac{D_m^2 \zeta_m}{D^2 I_m}$
per aggregation round $\tau$,
as stated in Corollary~\ref{rema-bound}, problem \textbf{P1} can be rewritten in a much simpler form, as given by 
\begin{subequations}\label{eq-P2}
	\begin{align}
	\textbf{P2}:\;\min_{\{\bm{\lambda} \}}\;\; 
&   \sum_{m \in {\cal M}}\frac{D_m^2 \zeta_m}{D^2 I_m} \label{eq-P2 a} \\ 
	{\text{s.t.}}\; 
 % & \zeta_m = {\rm sgn}\left(\sum_{k=1}^{K} \sum_{l=0}^{L} \lambda_{m,k,l}\right),
 % \label{eq-P2 bb} \\
%  & P_m \leq P^{\max}_m,\forall m,\label{eq-P2 b}\\
% 	&\sum_{k=1}^{K} \sum_{m = 1}^M\sum_{l=0}^{L} \lambda_{m,k,l}   \leq K,\label{eq-P2 c}\\
% 	&\sum_{m = 1}^M\sum_{l=0}^{L} \lambda_{m,k,l} \leq 1, ~\forall k, \label{eq-P2 d}\\
% 	& \lambda_{m,k,l} \in \{0,1\}, \label{eq-P2 e}\\
% % 	& R^{(1)}_m\left(\bm{\eta}\right) = \chi R^{(2)}_m\left(\bm{\eta}\right), \label{eq-P1 g}\\
% 	& T_m \leq T_{\rm th}, \label{eq-P2 f}\\
       & 1 \leq I_m \leq A, \label{eq-P2 g} \\
       & \eqref{eq-P1 b1} - \eqref{eq-P1 h}.\nonumber
	\end{align}
\end{subequations}
% where ${\rm sgn} (x)$ is the sign function. 
% $\zeta_m$ yields
% \begin{subequations}\label{eq.zeta}
%     \begin{numcases}{\zeta_m =}
%     & 1, ~\text{if}~\sum_{k=1}^{K}\sum_{l=0}^{L} \lambda_{m,k,l}   \geq 1; \label{eq.zeta.a} \\
%     & 0, ~\text{if}~\sum_{k=1}^{K}\sum_{l=0}^{L} \lambda_{m,k,l}   = 0. \label{eq.zeta.b} 
%     \end{numcases}
% \end{subequations}
% % Here, $\varrho_m$ and $\beta_m$ are constants and can take different values for different users.
% Here, \eqref{eq.zeta.a} indicates that client $m$ is selected if at least one subchannel and modulation are allocated to the client.
% \eqref{eq.zeta.b} indicates otherwise.

Problem \textbf{P2} is still difficult to tackle, as it is a mixed integer program.
% variable $\bm{\lambda}$, 
Moreover, the feasible solution region of the problem is non-convex since the optimization variable $\bm{\lambda}$
is discrete.
In the ensuing sections, we first convert the objective function to an explicit function of $\bm{\lambda}$,
and then obtain the almost surely optimal solution to the problem.

It is noted that in Problem \textbf{P2}, the number of training iterations per client per round, $I_m$, is also optimized while maximizing the objective since it can be uniquely determined given $T_m^{\rm C}$ and $T_m^{\rm UL}$ (or $\bm{\lambda}$).
This indicates that the selections of clients, channels and modulations are based on their available computing powers. In other words, OFDMA-F$^2$L is flexible in  the number of training iterations accomplished by each selected client by taking the potentially different and even time-varying computational capabilities of the clients into account.
This flexibility enables the algorithm to achieve better efficiency in both communication and computation.

\subsection{Problem Transformation}\label{sec. problem transformation}
We proceed to show that optimizing the local iteration numbers $I_m$ can be achieved
by optimizing the uplink transmission delay $T_m^{\rm UL}$.
From \eqref{eq.tmc} and \eqref{eq.tmm}, we have
\begin{subequations}
     \begin{align}
      I_m &= T_m^{C} \beta_m/\mu, \label{eq.emtau} \\
      T^{C}_m &= T_m - T^{\rm UL}_m - T^{\rm DL}. \label{eq.tmcc}
      \end{align}
\end{subequations}
Substituting \eqref{eq.tmcc} into \eqref{eq.emtau}, we have
\begin{equation}\label{eq-I_m}
    I_m = (T_m - T^{\rm UL}_m - T^{\rm DL}) \beta_m/\mu.
\end{equation}
Here, $I_m$ increases with $T_m$; $T^{\rm DL}$ and $\mu$ are given.

For the set of selected clients, i.e., $\cal M' \in \cal M$ ($|{\cal M'}| = M'$), we have $\zeta_{m} = 1, \,\forall m \in \cal M'$. Then, the objective function in \eqref{eq-P2 a} 
can be recast as $\sum_{m \in {\cal M'}} \frac{D_{m}^2}{D^2 I_{m}}$. 
Exploiting the harmonic-arithmetic mean inequality (i.e., $\frac{n}{\sum^n_{i=1} \frac{1}{x_i}} \leq \frac{\sum^n_{i=1} {x_i}}{n},\,\forall x_i>0$), it follows that
\begin{subequations}
\begin{align}
\sum_{m \in {\cal M'}} \frac{D_{m}^2}{D^2 I_{m}} & \geq 
    \frac{M'^2}{\sum_{m \in {\cal M'}} \frac{D^2 I_{m}}{D_{m}^2}}  \label{eq-obj 2a}\\
    & = \frac{M'^2}{\sum_{m \in {\cal M'}} {(T_{m} - T^{\rm UL}_{m} - T^{\rm DL}) \beta_{m}D^2/\mu D_{m}^2}} \label{eq-obj 2b}\\
    & \geq \frac{M'^2}{\sum_{m \in {\cal M'}} {(T_{\rm th} - T^{\rm UL}_{m} - T^{\rm DL}) \beta_{m}D^2/\mu D_{m}^2}},\label{eq-obj 2c}
\end{align}
\end{subequations}
where \eqref{eq-obj 2b} is obtained by plugging \eqref{eq-I_m} into \eqref{eq-obj 2a}, and 
\eqref{eq-obj 2c} is obtained by substituting \eqref{eq-P1 h} into \eqref{eq-obj 2b}.

We can minimize
$\sum_{m \in {\cal M'}} \frac{D_{m}^2}{D^2 I_{m}}$
by minimizing its lower bound, i.e.,~\eqref{eq-obj 2c},
% \begin{equation}
%     \min_{\{\bm{\lambda} \}} \frac{M'^2}{\sum_{m' \in {\cal M'}} \frac{I_{m'}}{\varrho_{m'}^2}},
% \end{equation}
which, in turn, is equivalent to maximizing the denominator of \eqref{eq-obj 2c}, i.e., 
$\sum_{m \in {\cal M'}} {(T_{\rm th} - T^{\rm UL}_{m} - T^{\rm DL}) \beta_{m}D^2/\mu D_{m}^2}$.
Since $D_{m}$, $D$, $T_{\rm th}$, $T^{\rm DL}$, and $\mu$ are given constants,
we only need to minimize $\sum_{m \in {\cal M'}} {T^{\rm UL}_{m} \beta_{m}D^2/D_{m}^2}$, or $\sum_{m \in {\cal M'}}  \beta_{m}D^2/R^{\rm UL}_{m}D_{m}^2$ based on~\eqref{eq.tmurmu}. By applying the harmonic-arithmetic mean inequality again,  
\begin{align}
    \sum_{m \in {\cal M'}}  \frac{\beta_{m}D^2}{R^{\rm UL}_{m}D_{m}^2}\geq \frac{M'^2}{ \sum_{m \in {\cal M'}}  \frac{R^{\rm UL}_{m}D_{m}^2}{\beta_{m}D^2} },
\end{align}
and 
minimizing $\sum_{m \in {\cal M'}}  \beta_{m}D^2/R^{\rm UL}_{m} D_{m}^2$ can be done by minimizing its minimum, i.e., $\frac{M'^2}{ \sum_{m \in {\cal M'}}  \frac{R^{\rm UL}_{m}D_{m}^2}{\beta_{m}D^2}}$, or maximizing 
$ \sum_{m \in {\cal M'}}  \frac{R^{\rm UL}_{m}D_{m}^2}{\beta_{m}D^2}$. 
Then, problem \textbf{P2} 
can be rewritten as 
\begin{subequations}\label{eq-P3}
	\begin{align}
	\textbf{P3}:\;\max_{\{\bm{\lambda} \}}\;\; 
&   \sum_{m \in {\cal M}}  \frac{R^{\rm UL}_{m}D_{m}^2\zeta_m}{\beta_{m}D^2} \label{eq-P3 a1}\\ 
	{\text{s.t.}}\; & T_{\rm th} - T^{\rm DL} - \frac{A\mu}{\beta_{m}} \le T^{\rm UL}_{m} \le T_{\rm th} - T^{\rm DL} - \frac{\mu}{\beta_{m}},\forall {m},\label{eq-P3 b}\\
& \eqref{eq-P1 b1} - \eqref{eq-P1 f}, \notag
	\end{align}
\end{subequations}
where \eqref{eq-P3 b} is obtained by substituting \eqref{eq-I_m} into \eqref{eq-P2 g}. Here, the objective in \eqref{eq-P3 a1} is equivalent to $ \sum_{m \in {\cal M'}}  \frac{R^{\rm UL}_{m}D_{m}^2}{\beta_{m}D^2}$ since $\zeta_m$ indicates the selection of clients; i.e., $\zeta_m=1$ if client $m$ is selected, or $\zeta_m=0$, otherwise.

Problem \textbf{P3} can be further rewritten as
\begin{subequations}\label{eq-P4}
	\begin{align}
	\textbf{P4}:\;\max_{\{\bm{\lambda} \}}\;\; 
% &   \sum_{m \in {\cal M}} \frac{R^u_m \varrho_m^2 \zeta^{\tau}_m }{\beta_m N(\bm{\omega}_m)} \label{eq-P4 a} \\ 
&   \sum_{m \in {\cal M}}  \frac{R^{\rm UL}_{m}D_{m}^2\zeta_m}{\beta_{m}D^2} \label{eq-P4 a} \\ 
   {\text{s.t.}}\; &\frac{N}{T_{\rm th} - T^{\rm DL} - \frac{\mu}{\beta_{m}}} \le R^{\rm UL}_{m} \le \frac{N}{T_{\rm th} - T^{\rm DL} - \frac{A\mu}{\beta_{m}}} ,\forall {m},
   \label{eq-P4 b} \\
   % & R^{\rm UL}_{m} \le \frac{N}{T_{\rm th} - T^{\rm DL} - \frac{A\mu}{\beta_{m}}} ,\forall {m},\label{eq-P4 c}\\
& \eqref{eq-P1 b1} - \eqref{eq-P1 f}, \notag
	\end{align}
\end{subequations}
where \eqref{eq-P4 b} is obtained by substituting~\eqref{eq.tmurmu} into~\eqref{eq-P3 b}.

From \eqref{eq.datarate} and \eqref{eq.zeta}, it follows that $R_m^{\rm UL}=0$
when $\zeta_m=0$,
and $R_m^{\rm UL}=\sum_{k = 1}^{K}\sum_{l=0}^{L} \lambda_{m,k,l} \times r_l$ when $\zeta_m=1$.
As a result, $R_m^{\rm UL} \zeta_m$ can be replaced by $R_m^{\rm UL}$, and \eqref{eq-P1 b1} can be suppressed.
As a result, problem \textbf{P4} is equivalent to
\begin{equation}\label{eq-P5}
\begin{aligned} 
   \textbf{P5}:\;  \max_{\{\bm{\lambda} \}}~ & \sum_{m \in {\cal M}} R^{\rm UL}_m D_m^2 /\beta_m D^2\\
{\text{s.t.}}\; & \eqref{eq-P1 b} - \eqref{eq-P1 f},~ \eqref{eq-P4 b}.
\end{aligned}
\end{equation}
Problem \textbf{P5} is in the form of weighted sum rate maximization, where the data rate of each client $m$ is weighted by the weighting coefficient $D_m^2 /\beta_m D^2$. 
However, the problem is still intractable for classic convex optimization tools, because it is a mixed-integer program as the result of constraint~\eqref{eq-P1 f}.
To solve problem \textbf{P5}, we propose a Lagrange-dual based method, as delineated in the next subsection.

{\blue Note that 
Problem \textbf{P5} has a form of weighted sum rate maximization,
where the weighting coefficient is $D_m^2 /\beta_m D^2$ for each client $m$.
This is reasonable because, given the delay threshold $T_{\rm th}$, the maximum number of local iterations that a client can run is bounded by its data rate in an aggregation round, as indicated in \eqref{eq-P4 b};
in other words, a higher data rate results in a shorter uplink delay and, in turn, a longer time for local training.
To this end, to maximize the number of local iterations per aggregation round would require the data rate to be maximized.
Moreover, the optimality gap of the FL is minimized by maximizing the number of local iterations per aggregation round; see (17c).
As a result, the minimization of the optimality gap can be transformed to the maximization of the weighted
sum rate, as stated in Problem \textbf{P5}.
% the weighted sum rate can be maximized to minimize the optimality gap.
}

\subsection{Optimal Client, Subchannel, and Modulation Selection}
By defining ${\bm{\xi}} = \{\xi_m,\forall m\}$, 
${\bm{\nu}} = \{\nu_m,\forall m\}$, and ${\bm{\iota}} = \{\iota_m,\forall m\}$ as the dual variables concerning \eqref{eq-P1 b} and \eqref{eq-P4 b}, respectively,
the Lagrange function of~\textbf{P5} is obtained as
\begin{equation}\label{eq-lagrange}
\begin{aligned}
    L\left(\bm{\lambda}, \bm{\xi}, {\bm{\nu}}, {\bm{\iota}}  \right)  = & \sum_{m=1}^{M} \left[ R^{\rm UL}_m D_m^2 /\beta_m D^2 - \xi_m \left( P_m - P_m^{\max} \right) \right] \\
    & - \sum_{m=1}^{M}\nu_m\left(- R^{\rm UL}_m + \frac{N}{T_{\rm th} - T^{\rm DL} - \frac{\mu}{\beta_m}}\right)\\
    & - \sum_{m=1}^{M}\iota_m\left(R^{\rm UL}_m - \frac{N}{T_{\rm th} - T^{\rm DL} - \frac{A\mu}{\beta_m}}\right).
\end{aligned}	
\end{equation}
Therefore, the dual problem is
% and the dual problem of \eqref{eq-P5} is
	\begin{equation}\label{eq-dual}
	\min_{{\bm{\xi}}, {\bm{\nu}}, {\bm{\iota}}} D\left(\xi, {\bm{\nu}}, {\bm{\iota}} \right).
	\end{equation}
The Lagrange dual function is
\begin{equation}
    D({\bm{\xi}}, {\bm{\nu}}, {\bm{\iota}}) = \max_{\{\bm{\lambda} \}} L\left(\bm{\lambda}, {\bm{\xi}}, {\bm{\nu}}, {\bm{\iota}} \right).
\end{equation}
By defining $\varpi_{m,k,l}\left(\xi_m,\nu_m,\iota_m \right) = 
	- \xi_m p_{m,k,l}\left(|h_{m,k}|^2\right) + \left(D_m^2 /\beta_m D^2 + \nu_m - \iota_m \right) r_l$
for brevity,
% \begin{equation}\label{eq-varpi}
% \begin{aligned}
% 	\varpi_{m,k,l}\left(\xi_m,\nu_m,\iota_m \right) = 
% 	- \xi_m p_{m,k,l}\left(|h_{m,k}|^2\right) + \left(\varrho_m^2 /\beta_m + \nu_m - \iota_m \right) r_l,
% \end{aligned}
% \end{equation}
\eqref{eq-lagrange} is rearranged as
\begin{equation}
	\begin{aligned}
	L\left(\bm{\lambda}, {\bm{\xi}}, {\bm{\nu}}, {\bm{\iota}} \right)
	& = \sum_{m = 1}^M \left(\xi_m P_m^{\max} \!+\! \sum_{k = 1}^K \sum_{l=0}^{L}\! \lambda_{m,k,l} \varpi_{m,k,l}\left(\xi_m,\nu_m, \iota_m \right)\right)\\
	& \;+ \sum_{m=1}^{M}\left(\frac{ \iota_m N}{T_{\rm th} - T^{\rm DL} - \frac{A\mu}{\beta_m}} - \frac{\nu_m N}{T_{\rm th} - T^{\rm DL} - \frac{\mu}{\beta_m}}\right).
	\end{aligned}
\end{equation}
% Then, the Lagrange dual function is
% \begin{equation}
%     D({\bm{\xi}}, {\bm{\nu}}, {\bm{\iota}}) = \max_{\{\bm{\lambda} \}} L\left(\bm{\lambda}, {\bm{\xi}}, {\bm{\nu}}, {\bm{\iota}} \right), 
% \end{equation}
% and the dual problem of \eqref{eq-P5} is
% 	\begin{equation}\label{eq-dual}
% 	\min_{{\bm{\xi}}, {\bm{\nu}}, {\bm{\iota}}} D\left(\xi, {\bm{\nu}}, {\bm{\iota}} \right).
% 	\end{equation}
Given ${\bm{\xi}}$, ${\bm{\nu}}$, and ${\bm{\iota}}$, the primary variable $\bm{\lambda} $
 can be acquired via resolving 
	\begin{equation}\label{eq.lag.lamb}
	\begin{aligned}
	\max_{\bm{\lambda}}\! & \sum_{k = 1}^K \!\Bigg\lbrace \!\sum_{m = 1}^M \sum_{l=0}^{L}  \lambda_{m,k,l} \varpi_{m,k,l}\left(\xi_m,\nu_m, \iota_m  \right)\!\Bigg\rbrace,\,{\text{s.t.}} \,\eqref{eq-P1 e},~\eqref{eq-P1 f}. 
	\end{aligned}	
	\end{equation}
	The optimal channel and modulation selections follow a ``winner-takes-all'' policy~\cite{He2014Optimal, yuan23}. For subchannel $k$, the $m_k^*$-th client is selected to upload its local model utilizing the $l_k^*$-th modulation, that is,
	\begin{equation}\label{eq-pair}
	\left\lbrace m^{\ast}_k, l^{\ast}_k \right\rbrace = \arg \max_{m,l}\!\; \varpi_{m,k,l}\left(\xi_m,\nu_m, \iota_m \right),\; \forall k \in {\cal K}.
	\end{equation}	
	A greedy policy is employed to optimize $\bm{\lambda}$:
	\begin{equation}\label{eq-opt-eta}
	\left\lbrace
	\begin{aligned}
	&\lambda^{\ast}_{m,k,l}  = 1,\, {\rm{if}}\, \left\lbrace m, l \right\rbrace = \left\lbrace m^{\ast}_k, l^{\ast}_k \right\rbrace;\\
	&\lambda^{\ast}_{m,k,l}  = 0,\, {\rm{otherwise}}.
	\end{aligned}\right.	 
	\end{equation}
	
	With ${\bm{\lambda}^{\ast}\left({\bm{\xi}}, {\bm{\nu}},  {\bm{\iota}} \right)}$ acquired in \eqref{eq-opt-eta}, the sub-gradient descent algorithm could be employed to refresh ${\bm{\xi}}$, ${\bm{\nu}}$ and ${\bm{\iota}}$ via solving~\eqref{eq-dual}. 
	%$\bm{g}\left({\bm{\eta}}^{\ast}(\lambda, {\bm{\nu}}) \right) = P_{\max} - P\left({\bm{\eta}} \right)$ is the sub-gradient of $D(\lambda, {\bm{\nu}})$~\cite{bazaraa2013nonlinear}, and 
	Then, ${\bm{\xi}}$, ${\bm{\nu}}$ and ${\bm{\iota}}$ are refreshed by~\cite{boyd2004convex}
	\begin{subequations}\label{eq-sub-gradient}
		\begin{align}
		\xi_m\left(\ell +1 \right) & = \left[ \xi_m(\ell ) + \varepsilon\left(P_m\left({\bm{\lambda}}^{\ast} \left({\bm{\xi}}(\ell ),{\bm{\nu}}(\ell ),{\bm{\iota}}(\ell ) \right) \right) - P_m^{\max}\right) \right]^+;\\
		\nu_m\left(\ell +1 \right) 
		& = \Bigg[ \nu_m (\ell ) + \varepsilon \Bigg(R^{\rm UL}_m\left({\bm{\lambda}}^{\ast}
   \left({\bm{\xi}}(\ell ),{\bm{\nu}}(\ell ),{\bm{\iota}}(\ell ) \right) \right) \notag \\ 
   & \qquad \qquad \qquad  - \frac{N}{T_{\rm th} - T^{\rm DL} - \frac{\mu}{\beta_m}}\Bigg) \Bigg]^+,\; \forall m;\\
		\iota_m\left(\ell +1 \right) &= \Bigg[ \iota_m (\ell ) + \varepsilon \Bigg( - R^{\rm UL}_m\left({\bm{\lambda}}^{\ast} \left({\bm{\xi}}(\ell ),{\bm{\nu}}(\ell ),{\bm{\iota}}(\ell ) \right) \right) \notag \\
  & \qquad \qquad \qquad + \frac{N}{T_{\rm th} - T^{\rm DL} - \frac{A\mu}{\beta_m}}\Bigg) \Bigg]^+,\; \forall m,
		\end{align}	
	\end{subequations}
	with $\varepsilon$ being the step size, $\ell$ being the index of an iteration, and $\left[x \right]^+ = \max\left(0,x \right)$.
 Initially, ${\bm{\xi}}$, $\bm{\nu}$ and $\bm{\iota}$ are no smaller than zero, that is, $\xi_m(0) \geq 0$, $\nu_m(0) \geq 0$, and $\iota_m(0) \ge 0, \forall m$,
 to ensure the convergence of \eqref{eq-sub-gradient}.

\begin{algorithm}[t]\small
	\DontPrintSemicolon
	\caption{Proposed Optimal Client, Subchannel, and Modulation Selection.}
	\label{algo_TD3}
	{\bf Initialization:}         
        {
Input ${\bm{\xi}}(0)$, ${\bm{\nu}}(0)$, ${\bm{\iota}}(0)$, $\ell=0$,
and the maximum iteration number of the Lagrange method~$L_{\max}$. \\
Obtain the dual problem of \textbf{P5} according to \eqref{eq-dual}: $\min_{{\bm{\xi}}, {\bm{\nu}}, {\bm{\iota}}}\max_{\bm{\lambda}} L\left(\bm{\lambda}, {\bm{\xi}}, {\bm{\nu}}, {\bm{\iota}} \right)$.\\
\While{$L\left(\bm{\lambda}, {\bm{\xi}}, {\bm{\nu}}, {\bm{\iota}} \right)$ does not converge, and $\ell \leq L_{\max}$}{
Obtain $\bm{\lambda}^\ast (\ell+1)$ by maximizing $L\left(\bm{\lambda}, {\bm{\xi}}(\ell), {\bm{\nu}}(\ell), {\bm{\iota}}(\ell) \right)$ using a greedy strategy as in \eqref{eq.lag.lamb}--\eqref{eq-opt-eta}.\\
% \While{$P\left({\bm{\lambda}^\ast}\right)$ is yet to converge, and $J \leq J_{\max}$}{
Update ${\bm{\xi}}(\ell+1)$, ${\bm{\nu}}(\ell+1)$ and ${\bm{\iota}}(\ell+1)$ according
to \eqref{eq-sub-gradient} with $\bm{\lambda}^\ast (\ell+1)$.\\
$\ell \leftarrow \ell + 1$.}		
% $I \leftarrow I + 1$.}   
}
Output the optimal selection of subchannels and modulations ${\bm \lambda} = \bm{\lambda}^{\ast}$.\\ 
		% }
\end{algorithm}

While solving \eqref{eq-P5} involves a non-convex mixed-integer programming problem, the zero-duality gap can still be ensured when the random channel gain $|h_{m,k}|^2, m\in {\cal M},\; k\in {\cal K}$, has a differentiable cumulative distribution function (CDF), as stated in the following lemma.

\begin{lemma}\label{prop}
	Given any ergodic fading with a continuous CDF, task~\textbf{P5} holds with strong duality. 
 The solution that is almost surely optimal (i.e., with probability of one) for
   % The almost surely optimal solution to 
   task \textbf{P5} %( with probability of 1)
   is $\lambda^{\ast}_{m,k,l} \left( \xi^{\ast}_m, \nu^{\ast}_m, \iota^{\ast}_m \right), \forall m,k$,
    when the channel gains $|h_{m,k}|^2$ have a continuous CDF,
    where $\xi^{\ast}_m$, $ \nu^{\ast}_m$, and $\iota^{\ast}_m$ are obtained in \eqref{eq-sub-gradient} with any initial $\xi_m(0) > 0$, $\nu_m(0) \geq 0$, and $\iota_m(0) \ge 0, \forall m$~\cite{GaoOptimal2011}.
\end{lemma}

\begin{proof}
	The almost sure optimality of the solution to problem~\textbf{P5} can be established by first ratifying the almost sure uniqueness of the solution in the following three situations. 
	
	First, if $\max_{m,l} \varpi_{m,k,l}\left(\xi_m,\nu_m,\iota_m \right) =0$, then all clients’ channels experience a deep fade in subchannel $k$ utilizing modulation $l$. 
	Even though client $m$ is chosen for the subchannel, $l^{\ast}_k(\xi_m,\nu_m,\iota_m) = 0$, the BS should refrain from transmitting on subchannel $k$; see \eqref{eq-opt-eta}. 
	% The best choice for the BS is to refrain from transmitting in the $k$-th subchannel, as given in \eqref{eq-opt-eta}. 
	The ``winner'' of the subchannel is assigned with the 0-th modulation mode (i.e., no transmission). 
	Second,	if $\max_{m,l} \varpi_{m,k,l}$ $(\xi_m,\nu_m,\iota_m) > 0$ and a sole ``winner'' obtains the subchannel, then there is a unique optimal policy from \eqref{eq-opt-eta}. 
	Third, if $\max_{m,l} \varpi_{m,k,l}\left(\xi_m,\nu_m,\iota_m\right) > 0$ and several $\{m,l\}$ pairs gain the $k$-th subchannel, the policy becomes non-unique. 
	Nevertheless, having multiple ``winners'' accounts for an event of Lebesgue measure zero~\cite{lenz2002singular}, given a continuous CDF of the stochastic channel gain. 
	The selection of a ``winner'' is subject to a ``measure zero'', as we maximize the average net reward in \eqref{eq-pair}. 
	By considering all three cases, the optimal allocation strategy $\lambda^{\ast}_{m,k,l} \left( \xi_m,\nu_m,\iota_m \right)$ in \eqref{eq-opt-eta} is unique with probability of~1.	
	
	Given the almost sure uniqueness, we confirm the (almost sure) optimality of the solution
    to problem \textbf{P5} obtained by running \eqref{eq-opt-eta} and \eqref{eq-sub-gradient}. 
	By relaxing $\lambda_{m,k,l} $ to a continuous value within $[0, 1]$, task \textbf{P2} becomes a convex linear program (LP).
    The dual task in \eqref{eq-dual} still holds for the relaxed LP. 
	Since the LP enjoys strong duality, the ``winner-takes-all'' policy ${\bm{\lambda}}^{\ast}(\bm{\xi},\bm{\nu},\bm{\iota})$ maximizing the dual task of the LP is almost surely unique and hence optimal for the LP.
    As the LP offers no smaller maximum objective value than problem \textbf{P5}, 
    ${\bm{\lambda}}^{\ast}(\bm{\xi},\bm{\nu},\bm{\iota})$ is almost surely optimal for~\eqref{eq-P5}.	 
\end{proof}

{\blue 
While future OFDMA systems have the potential to flexibly assign subchannels to a device (as done in this paper),
earlier systems, such as 3GPP Long-Term Evolution (LTE) and IEEE 802.11ax, may have more rigid constraints for subchannel allocation. 
For example, a device can only be assigned with up to two separate clusters of consecutive subchannels
in LTE/LTE-Advanced~\cite{xiaojing}.
The proposed algorithm can be potentially extended to incorporate this constraint.
For example, when a client is assigned with three clusters of subchannels after $\boldsymbol{\lambda}$
is obtained in~(32),
we can keep two clusters for the client and redistribute the third cluster to the clients allocated with subchannels adjacent to the cluster or those allocated with fewer than two clusters of subchannels. The performance of this extension will be analyzed in our future work.
}

\subsection{Complexity Analysis}
The proposed algorithm for solving \textbf{P5} is summarized in Algorithm~\ref{algo_TD3}.
Algorithm~\ref{algo_TD3} has a computational complexity of $\mathcal{O}(KML\log(1/\epsilon))$,
with $\mathcal{O}(\log(1/\epsilon))$ representing iterations required to attain the accuracy requirement $\epsilon$.
During an iteration, \eqref{eq-opt-eta} is employed to evaluate all $M$ clients, and $L$ modulations for subchannel $k$ to select the 2-pair $\left\lbrace m^{\ast}_k, l^{\ast}_k \right\rbrace$ that maximizes the net reward 
$\varpi_{m,k,l}\left(\lambda, \nu_m \right)$, resulting in the complexity of ${\cal O}(ML)$. 
The complexity of the other steps, such as the dual variables updating in~\eqref{eq-sub-gradient},
is $\mathcal{O}(M+1)$, which is relatively negligible. 
Given $K$ subchannels, the total complexity of Algorithm~2 is $\mathcal{O}(KML\log(1/\epsilon))$.	
{\blue The complexity scales linearly with the number of clients.
This is comparable with the best-case scenario of the existing technique, where
% A few techniques can be potentially adopted to further reduce the complexity when the number of clients is excessively large.
the weighted sum rate maximization problem, i.e., Problem \textbf{P5},
is transformed into a bipartite matching problem and solved with the Hungarian algorithm.
The Hungarian algorithm can eliminate the need for gradient calculations and has a best-case complexity of $\mathcal{O}(KLM)$, but it may suffer from a worst-case complexity of $\mathcal{O}(KLM^2)$~\cite{chenmzz}.
% Another possible technique is to preclude the clients with the worst channel qualities
% and allocate the resources evenly among the selected clients.
% This method amounts to a sorting algorithm,
% and can substantially reduce the computational complexity to $\mathcal{O}(M\log M)$ or $\mathcal{O}(M)$,
% However, it only yields a suboptimal solution~\cite{goodrich}.

}

\section{Simulation Results}\label{sec-sim}	
A square region that has a side length of 250~m is considered, with the BS placed at its center, i.e., $(0, 0, 20)$~m. 
The sides of the region are aligned with the $x$- and $y$-axes. 
The edge clients are uniformly distributed within the square region, as depicted in Fig.~\ref{fig-sysmodel}. 
The position of the edge client $m$ is $(x_{m},y_{m}, 1.5)$ m, $\forall {m} \in {\cal M}$. 
We set $M = 10$ clients. The maximum number of iterations between two sequential aggregations is $A = 10$.
We consider Rayleigh fading for both uplink and downlink channels.
The BER requirement of each client is set to $\chi_0 = 10^{-6}$.
The maximum number of aggregations is $G = 50$ for the MLP model and $G = 300$ for the CNN model; 
unless otherwise specified.
The model size of the MLP and CNN are $1.018 \times 10^5$ and $4.214 \times 10^5$, respectively. 
% {\red We evaluate OFDMA-F$^2$L with a time-varying perturbation noise variance on the MLP and CNN models.} 
The other parameters are provided in Table~\ref{tab.pm}.
\begin{table}[t]%\small
	\caption{Parameter setting of the OFDMA-F$^2$L}
	\begin{center}
		\begin{tabular}{ll}
			%\hline
			\toprule[1.5pt]
			Parameters  & Values \\ \hline
%			Parameter  & Values \\ \hline
			The BS's maximum transmit power, $P_{\max}$    & 30 dBm \\
			Power budget of client $m$, $P_m^{\max}$    & 20 dBm \\ 
			Number of subchannels, $K$            & 8 or 16 \\ 
			Number of clients, $M$              & 10 \\ 
			Number of modulations, $L$   & 4 \\
			Modulation rate set & \{0,2,4,6\} bits/symbol\\
			Path loss at $d_0 =1$ m, $\epsilon_o$      &-30 dB \\ 
			Path loss exponent, $\alpha_{d}$   & 2.8 \\ 
% 			Rician factors, $K_1$, $K_2$, $K_3$          & 1, 3, 1     \\ 
			Noise power density, $\sigma^2$                 & -169 dBm/Hz \\ 
			Bandwidth, $B_w$                         & 100~MHz \\
			BER requirements, $\chi_0$  & $\{10^{-6}\}$ \\ 
			Coefficients of modulation and coding, $\beta_1$, $\beta_2$     &0.2, -1.6~\cite{Goldsmith1998} \\
			Model Size (MLP) &  $1.018\times 10^5$ \\
			Model Size (CNN) &  $4.214\times 10^5$ \\
% 			Maximum iterations between two consecutive aggregations, $A$ & 15\\
			Time duration of an aggregation, $T_{\rm th}$ & 10 seconds\\
			FLOPs required for training a local model, $\mu$ & 0.2 \\
			Computational speed of the edge client, $\beta$ & U(9,12)\\
			Local date size, $D_m$ & U(300, 500)\\
			Global data size, $D$ & $\sum_m D_m$\\
			\toprule[1.5pt]
		\end{tabular}
	\end{center}
	\label{tab.pm}
\end{table}

% {\color{red} \textbf{COMMENT:} What is the value of $\chi_0$? It needs to be very small, or we can be easily asked about transmission failures and retransmissions.}

We perform the experiments on three datasets:
\begin{itemize}
	\item \textit{Standard MNIST dataset}, comprised of 60,000 training and 10,000 testing grayscale images. The images depict handwritten digits ranging from one to ten; 
	
% 	\item The ADULT dataset, which contains 40,000 records extracted from census data and each record has up to 58 attributes, including age, education, etc.; 
	
	\item \textit{CIFAR10 dataset}, comprised of 60,000 color images of size $32 \times 32$, split into ten classes (6,000 per class). The dataset includes 50,000 and 10,000 images for training and testing, respectively;
	
	\item \textit{Fashion-MNIST (FMNIST) dataset}, comprised of Zalando's article images, which are $28 \times 28$ grayscale images classified into ten classes. It includes 60,000 and 10,000 images for training and testing, respectively. 
	Each example corresponds to a grayscale image of size  $28 \times 28$ and is labeled with one of the ten classes. 
\end{itemize}

{\blue
To the best of our knowledge, there is no existing study on the OFDMA-F$^2$L systems.
Let alone solutions to the challenging discrete selections of the clients, channels, and modulations in the system.
With due diligence, we come up with the following three baselines, which are the possible alternatives to the OFDMA-F$^2$L with the proposed optimal client, subchannel, and modulation selections.
\begin{itemize}
    \item {\textbf{Baseline~1:}} This is OFDMA-F$^2$L with only two modulations, where the clients, subchannels, and modulations are optimally selected in the same way as done in Algorithm~2.
    $l = 0$ indicates no transmission from a selected client to the BS.
    $l = 1$ indicates the client transmits to the BS with a modulation rate of 4 bits/symbol.
    
    \item {\textbf{Baseline~2:}} This is OFDMA-F$^2$L with random client selection and optimal subchannel and modulation selections. The optimal subchannel and modulation selections are done by following the
    ``winner-takes-all'' strategy.

    \item {\textbf{Baseline~3:}} 
    This scheme runs Sync-FL in the OFDMA system, referred to as OFDMA-Sync-FL, where the selections of the clients, channels, and modulations are jointly optimized under a persistent number of local iterations across all selected clients in each aggregation round.
    Specifically, by setting $I_m = A$ in problem \textbf{P2} when $\zeta_m=1$ and following the analysis in Section~\ref{sec. problem transformation}, problem~\textbf{P4} is updated for OFDMA-Sync-FL as 
    \begin{equation}
\begin{aligned}
	\textbf{P4'}:\;\max_{\{\bm{\lambda} \}}\;\; 
&   \sum_{m \in {\cal M}} D_m^2 \zeta_m /D^2 \label{eq-P3 a} \\ 
	{\text{s.t.}}\; & T^{\rm UL}_m \leq T_{\rm th} - T^{\rm DL} - \frac{A\mu}{\beta_m}, \\
& \eqref{eq-P1 b}-\eqref{eq-P1 f}, \notag
	\end{aligned}
    \end{equation}
    which can be solved using the ``winner-takes-all'' strategy, as in Algorithm~\ref{algo_TD3}.
    % The selections are optimized by running Algorithm~2 per aggregation round. The selected users upload their local models simultaneously after all of them have completed the same number of iterations.}}

\end{itemize}
}
All experiments are performed on a server running on Ubuntu 20.04.5 LTS, equipped with Intel@ Xeon(R) W-2235 CPU@3.80GHz featuring 12 cores and an NVIDIA GeForce RTX~3090.
% , and running on Ubuntu 20.04.5 LTS.

\subsection{OFDMA-F$^2$L on MLP Model}\label{subsec-mlp}
We first evaluate the impact of the optimal client, subchannel, and modulation selections on the learning performance (or convergence)
of OFDMA-F$^2$L on the MLP model.
The MLP is a feedforward neural network with fully-connected layers. 
% We use the MLP model containing a hidden layer with 32 units and train it on the MNIST dataset. 
The MNIST dataset is used to train the MLP with a 32-unit hidden layer.
We employ the softmax with ten classes regarding the ten digits. A cross-entropy loss function is used to capture the error of the model trained on the local dataset. 
The accuracy of the model is defined as the ratio of the number of correct predictions made by the model to the total number of predictions.
\begin{figure}[t]
	\centering
	\includegraphics[width=0.47\textwidth]{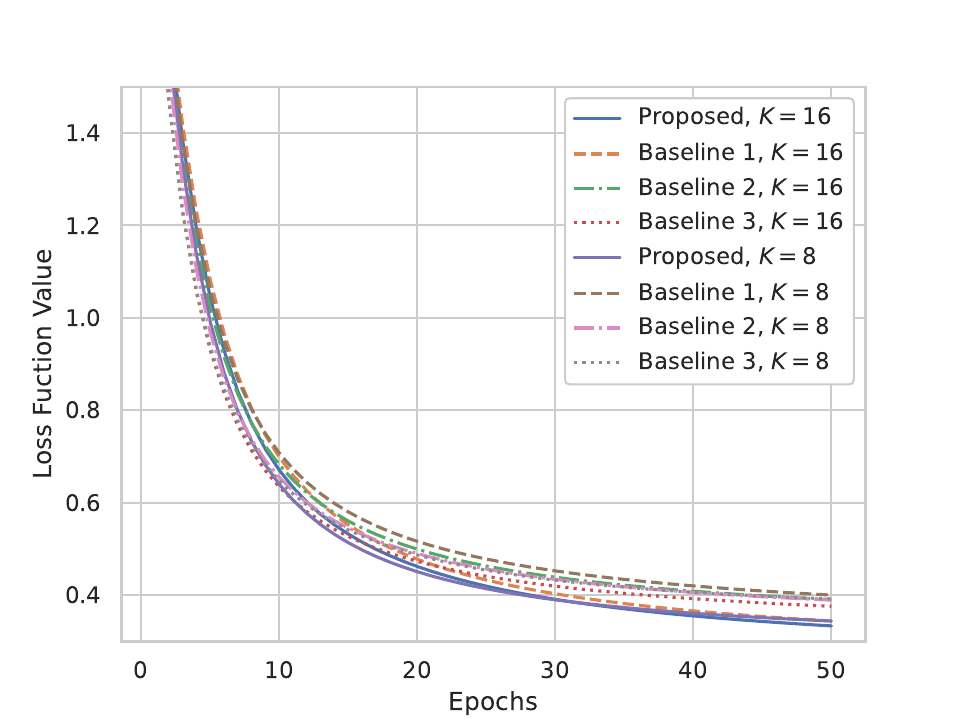}
	\caption{Loss of the MLP on the MNIST under $K = 8$ and $K=16$.}
	\label{fig.loss_mnist}
\end{figure}
\begin{figure}[t]
	\centering
	\includegraphics[width=0.47\textwidth]{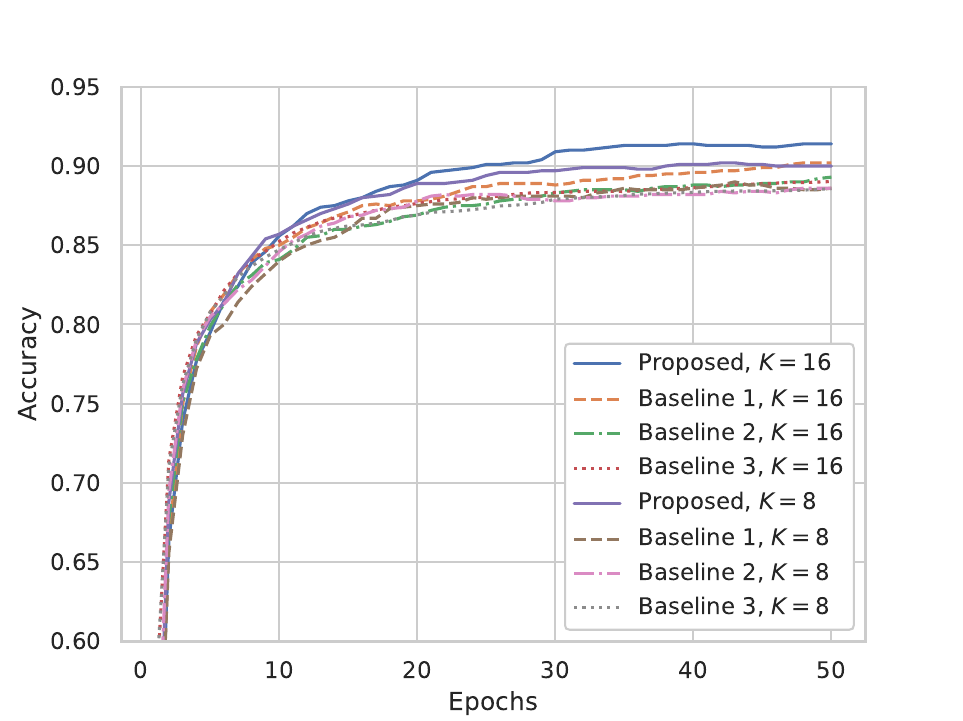}
	\caption{Accuracy of the MLP on the MNIST under $K = 8$ and $K=16$.}
	\label{fig.acc_mnist}
\end{figure}
% \begin{figure}[!t]
% \captionsetup{width=0.47\textwidth}
%     \centering
%         \begin{minipage}[t]{0.495\linewidth}
%         \centering
%         \includegraphics[width=\columnwidth]{./images/fig_loss.eps}
%         \caption{Loss of the MLP on the MNIST under $K = 8$ and $K=16$.}
% 	\label{fig.loss_mnist}
%     \end{minipage}
%     \hfill
%     \begin{minipage}[t]{0.495\linewidth}
%         \centering
%         \includegraphics[width=\columnwidth]{./images/fig_acc.eps}
%         \caption{Accuracy of the MLP on the MNIST under $K = 8$ and $K=16$.}
% 	\label{fig.acc_mnist}
%     \end{minipage}%
% \end{figure}

% We first compare the training performance of the proposed and baseline algorithms on the MLP and CNN models
% when there are $K = 8$ and $K=16$ subchannels.
Figs.~\ref{fig.loss_mnist} and \ref{fig.acc_mnist} show the (testing) loss (function value) and accuracy, respectively, of the MLP on the MNIST when there are $K = 8$ and $16$ subchannels. The $x$-axis indicates the global aggregations. 
It is seen that the FL algorithms show fast and smooth convergence within $50$ training epochs.
The loss function value reaches about $0.35$, and the accuracy is about $0.92$ under 
OFDMA-F$^2$L with the optimal selections.
Given the number of subchannels, i.e., $K = 8$ or $K=16$,
OFDMA-F$^2$L with the optimal selections outperforms the three baseline schemes with faster convergence speeds, smaller losses, and higher learning accuracies.
In addition, Baseline~1 is better than Baseline~2 since Baseline~1 optimizes the client selection policy while Baseline~2 selects clients randomly. The performances of Baselines~2 and 3 are close since they each execute part of the proposed OFDMA-F$^2$L with the optimal selections.

\subsection{OFDMA-F$^2$L on CNN Model}\label{subsec-cnn}
We also evaluate the impact of the optimal client, subchannel, and modulation selections on the learning performance (or convergence)
of OFDMA-F$^2$L on the CNN model. 
The CNN comprises two convolutional layers that have a kernel size of five, followed by three fully-connected layers. 
We train the CNN using the SGD optimizer.
% It is trained on the CIFAR10 and FMNIST datasets using the SGD optimizer. 
The ReLU units are used along with a softmax function that corresponds to ten classes for the CIFAR10 and ten digits of the FMNIST. 
% The CNN model is trained using the SGD.
% to minimize the loss function.
% \begin{figure}[!t]
% \captionsetup{width=0.47\textwidth}
%     \centering
%         \begin{minipage}[t]{0.495\linewidth}
%         \centering
%         \includegraphics[width=\columnwidth]{./images/fig_loss_cifar.eps}
% 	\caption{Loss of the CNN on the CIFAR10 under $K = 8$ and $K=16$.}
% 	\label{fig.loss_cifar}
%     \end{minipage}
%     \hfill
%     \begin{minipage}[t]{0.495\linewidth}
%         \centering
%         \includegraphics[width=\columnwidth]{./images/fig_acc_cifar.eps}
% 	\caption{Accuracy of the CNN on the CIFAR10 under $K = 8$ and $K=16$.}
% 	\label{fig.acc_cifar}
%     \end{minipage}%
% \end{figure}
\begin{figure}[!t]
	\centering
	\includegraphics[width=0.47\textwidth]{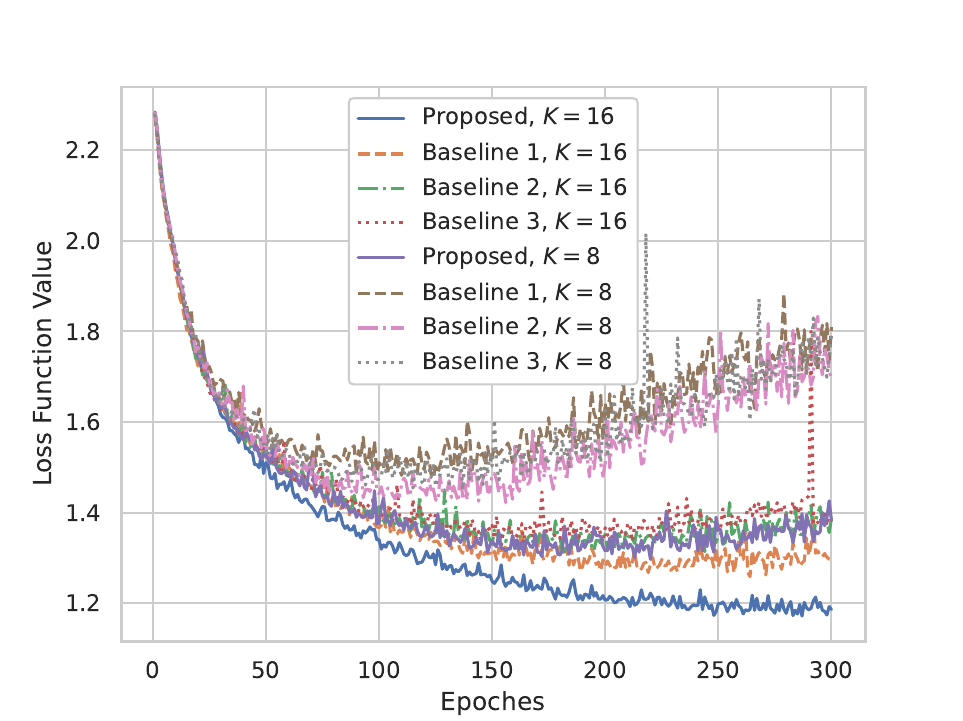}
	\caption{Loss of the CNN on the CIFAR10 under $K = 8$ and $K=16$.}
	\label{fig.loss_cifar}
\end{figure}
\begin{figure}[!t]
	\centering
	\includegraphics[width=0.47\textwidth]{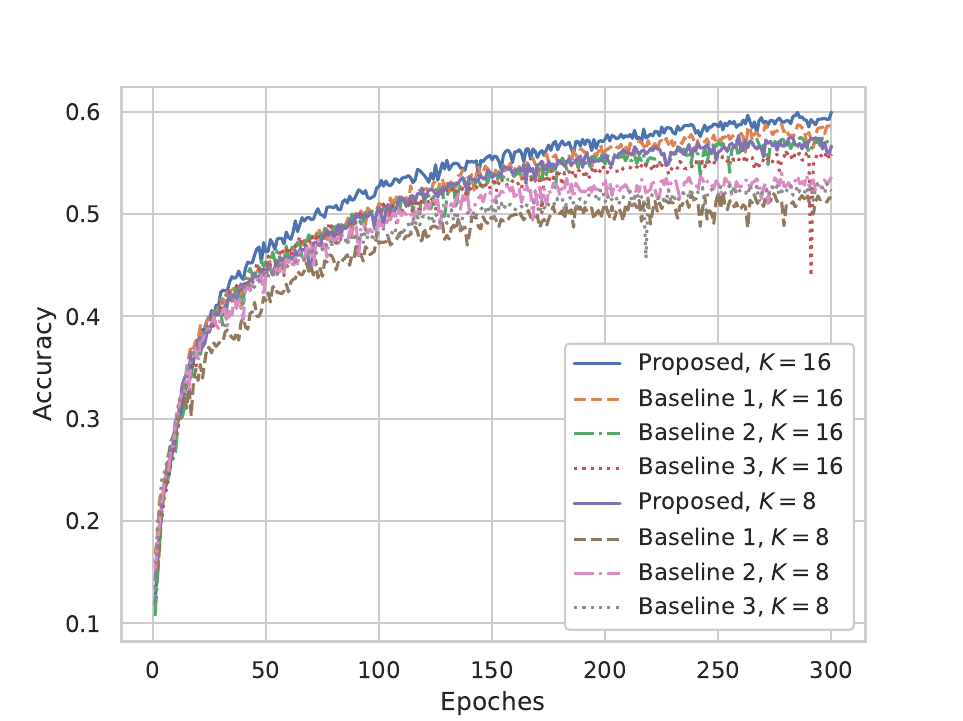}
	\caption{Accuracy of the CNN on the CIFAR10 under $K = 8$ and $K=16$.}
	\label{fig.acc_cifar}
\end{figure}

% \begin{figure}[!t]
% \captionsetup{width=0.47\textwidth}
%     \centering
%         \begin{minipage}[t]{0.495\linewidth}
%         \centering
%         \includegraphics[width=\columnwidth]{./images/fig_loss_fmnist.eps}
% 	\caption{Loss of the CNN on the FMNIST under $K = 8$ and $K=16$.}
% 	\label{fig.loss_fmnist}
%     \end{minipage}
%     \hfill
%     \begin{minipage}[t]{0.495\linewidth}
%         \centering
%         \includegraphics[width=\columnwidth]{./images/fig_acc_fmnist.eps}
% 	\caption{Accuracy of the CNN on the FMNIST under $K = 8$ and $K=16$.}
% 	\label{fig.acc_fmnist}
%     \end{minipage}%
% \end{figure}
\begin{figure}[t]
	\centering
	\includegraphics[width=0.47\textwidth]{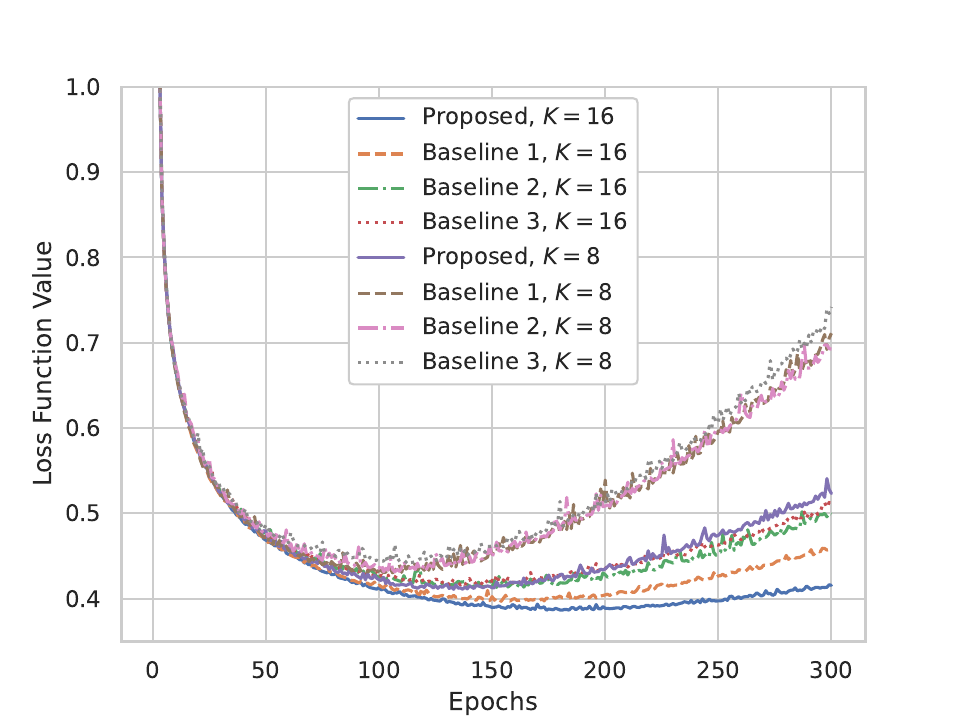}
	\caption{Loss of the CNN on the FMNIST under $K = 8$ and $K=16$.}
	\label{fig.loss_fmnist}
\end{figure}
\begin{figure}[t]
	\centering
	\includegraphics[width=0.47\textwidth]{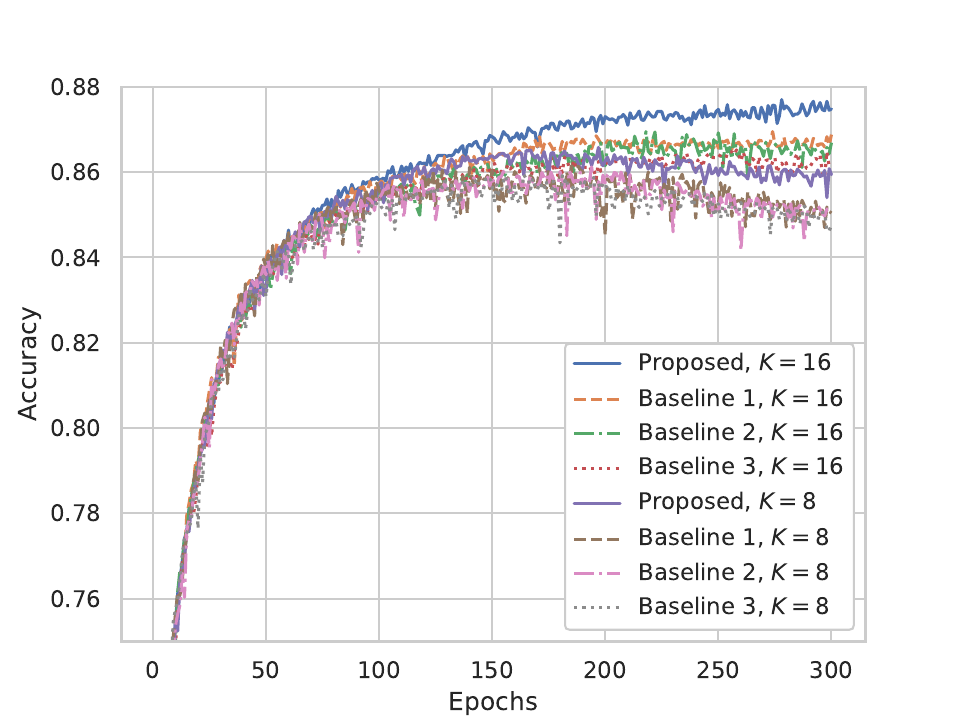}
	\caption{Accuracy of the CNN on the FMNIST under $K = 8$ and $K=16$.}
	\label{fig.acc_fmnist}
\end{figure}

Figs.~\ref{fig.loss_cifar}-\ref{fig.acc_fmnist} show the (testing) losses and accuracies
of the CNN 
on the CIFAR10 and FMNIST with $300$ training epochs.
As the training progresses, the testing loss decreases, and the testing accuracy increases.
At the beginning, the testing loss decreases rapidly,
indicating that the model learns fast from the available data.
Once a certain number of aggregations is reached, the rate of improvement slows down,
and the testing loss converges to the minimum value, as shown in Figs.~\ref{fig.loss_cifar} and~\ref{fig.loss_fmnist}.
Similarly, the testing accuracy increases rapidly at first, but levels off.
The model reaches a plateau in its performance in Figs.~\ref{fig.acc_cifar} and \ref{fig.acc_fmnist}.
% It is seen that the FL algorithms show slower ($300$ training epochs) and less smooth convergence on the CIFAR10 dataset than on the MNIST dataset.
% This is because the CIFAR10 dataset is more difficult to train,
% as it contains images with a higher resolution and more features.

It is observed in Figs.~\ref{fig.loss_cifar} and \ref{fig.acc_cifar} that on the CIFAR10 dataset, the loss function value stabilizes around $1.2$ and the accuracy is about $0.6$ under OFDMA-F$^2$L with the optimal selections. In contrast, the three baseline schemes yield higher loss function values and worse accuracies.
In particular, when $K = 8$, the (testing) {\red losses} of the three baselines first decrease and then increase as the training goes on,
indicating a divergence of the model learned from the training data
(or, in other words, the model incurs overfitting to the training data and does not generalize well
to the testing data).
This is because when $K$ is small (e.g., $K = 8$), the total transmission rate is low,
and the number of clients involved in the model aggregation is small, in which case, the baselines could not provide enough data for learning a good model.

It is also observed in Figs.~\ref{fig.loss_cifar} and \ref{fig.acc_cifar} that Baseline 3 is the least stable, since it necessitates all selected clients to upload their local models simultaneously after completing the same number of iterations using Sync-FL. In this case, the global training can be significantly affected by clients with poor channels or low computing powers, especially when a large dataset is concerned. 

% when the network size is small, and the data transmission rate is low.

In Figs.~\ref{fig.loss_fmnist} and \ref{fig.acc_fmnist}, we see that the FL algorithms are much smoother than on the CIFAR10 dataset and have faster convergence.
This is because the FMNIST is a dataset of images of clothing items,
which is considerably smaller compared to the CIFAR10 dataset.
The images in the FMNIST dataset are also simpler and have fewer variations in color and texture, compared to the CIFAR10 images. 
As a result, OFDMA-F$^2$L with the optimal selections may take a shorter time to converge to a satisfactory solution.
When $K=16$, the accuracy of OFDMA-F$^2$L with the optimal selections stabilizes at about $0.88$,
while the accuracies of the three baselines are about $0.865$.
In other words, OFDMA-F$^2$L with the optimal selections has better learning performance
than the baselines. 
When $K = 8$, the (testing) loss function values of all four algorithms initially decrease
and then increase during training, indicating that the models undergo overfitting. 
This is likely due to the low transmission rate and the limited number of clients involved in the training process, which are insufficient to capture the full complexity of the data distribution.

\subsection{Performance Analysis of OFDMA-F$^2$L}
\begin{figure}[t]
	\centering
	\includegraphics[width=0.47\textwidth]{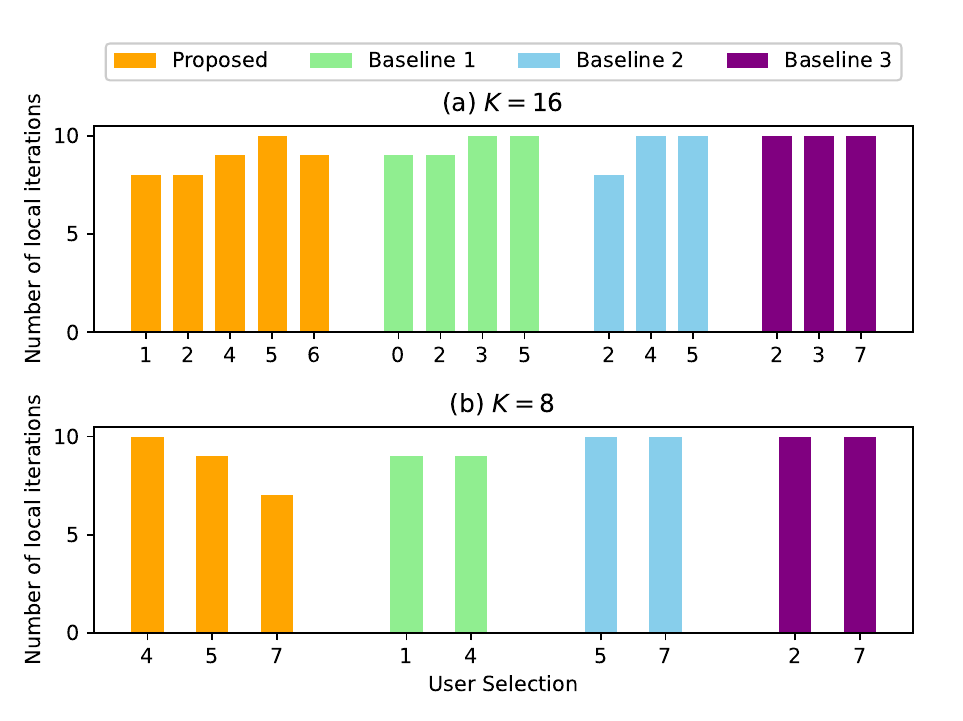}
	\caption{The number of local iterations in a selected aggregation round vs. the selected clients in the round.}
	\label{fig.user_selection}
\end{figure}
\begin{figure}[t]
	\centering
	\includegraphics[width=0.47\textwidth]{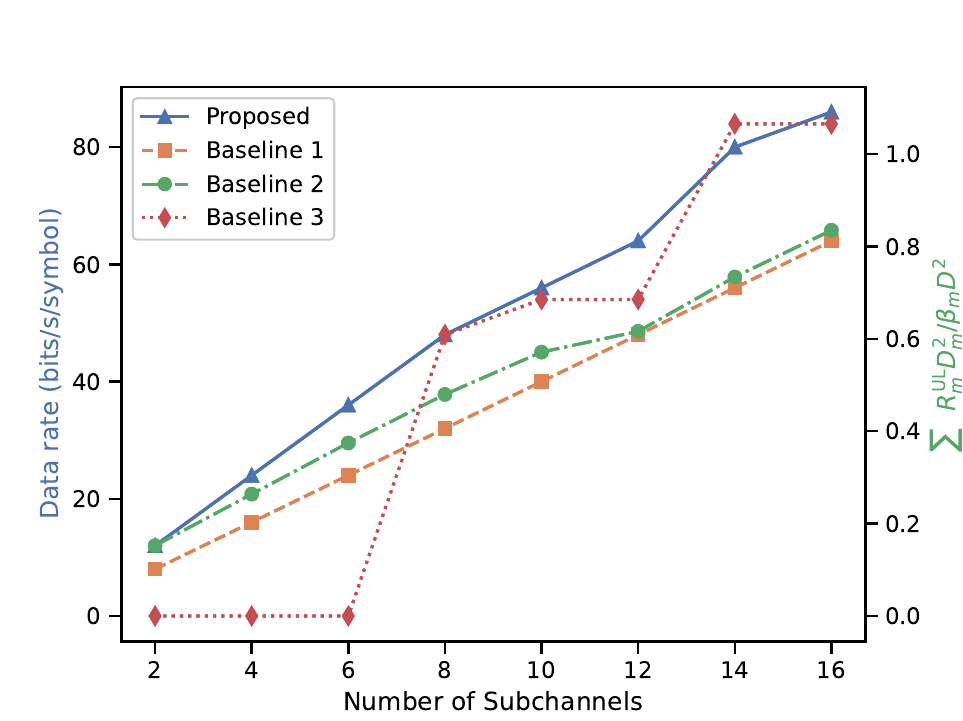}
	\caption{$\sum_{m \in {\cal M}} R^{\rm UL}_m \varrho_m^2 / \beta_m$ and sum data rate vs. the number of subchannels.}
	\label{fig.obj}
\end{figure}

% \begin{figure}[!t]
% \captionsetup{width=0.47\textwidth}
%     \centering
%         \begin{minipage}[t]{0.495\linewidth}
%         \centering
%         \includegraphics[width=\columnwidth]{./images/fig_bar.eps}
% 	\caption{The number of local iterations in a selected aggregation round vs. the selected clients in the round.}
% 	\label{fig.user_selection}
%     \end{minipage}
%     \hfill
%     \begin{minipage}[t]{0.495\linewidth}
%         \centering
%         \includegraphics[width=\columnwidth]{./images/fig_objective_rate_sub.eps}
% 	\caption{$\sum_{m \in {\cal M}} R^{\rm UL}_m D_m^2 / \beta_m D^2$ and sum data rate vs. the number of subcarriers.}
% 	\label{fig.obj}
%     \end{minipage}%
% \end{figure}
Fig.~\ref{fig.user_selection} plots the client selection results and corresponding numbers of local iterations
for each selected client at a specific training epoch of OFDMA-F$^2$L
and the other baselines when $K=8$ and $K=16$.
Compared to the three baselines, more clients are selected under OFDMA-F$^2$L with the optimal selections.
Besides, OFDMA-F$^2$L with the optimal selections accommodates more local iterations per aggregation round than the three baselines,
shows a better learning performance than the baselines in Sections~\ref{subsec-mlp} and~\ref{subsec-cnn}. 
When $K = 16$, Baseline 1 supports more local iterations per aggregation round than Baseline 2.
When $K = 8$, Baseline~1 accommodates fewer local iterations per aggregation round than Baseline 2. This is consistent with the learning performance of the three baselines presented in Sections~\ref{subsec-mlp} and~\ref{subsec-cnn}. 

\begin{figure}[t]
	\centering
	\includegraphics[width=0.47\textwidth]{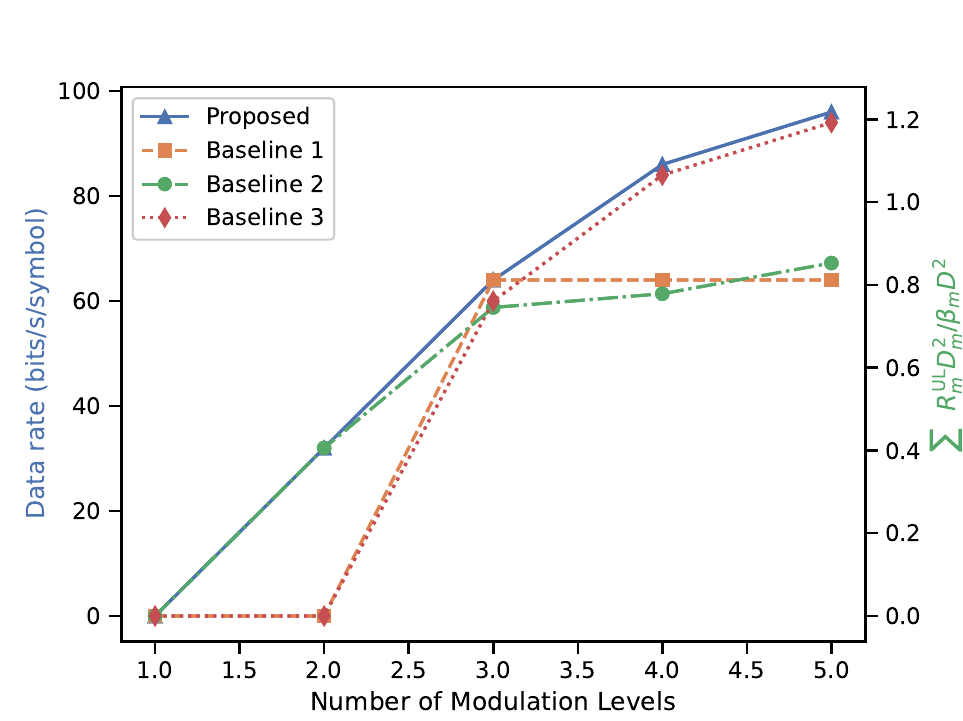}
	\caption{$\sum_{m \in {\cal M}} R^{\rm UL}_m D_m^2 / \beta_m D^2$ and sum data rate vs. the number of modulation modes.}
	\label{fig.obj2}
\end{figure}

Figs.~\ref{fig.obj} and \ref{fig.obj2} show the objective value of problem \textbf{P5}
and the (uplink) sum rate of the proposed OFDMA-F$^2$L with the optimal selections
when the number of subchannels or modulation modes changes.
In general, the performance metrics increase as the subchannels or modulation modes rise since the system can have a higher data rate, a shorter communication delay,
and hence the lower objective function value and higher accuracy.
OFDMA-F$^2$L with the optimal selections significantly outperforms the three baselines, as the subchannels or modulations increase.

It is also observed in Fig.~\ref{fig.obj} that the objective value and data rate change linearly
and synchronously with the number of subchannels under Baseline 1, as it has only two possible modulation modes.
When the number of subchannels $K \le 6$, the data rate is always zero under Baseline 3.
This is because the total delay constraint cannot be satisfied with insufficient communication resources.
It is further shown in Fig.~\ref{fig.obj2} that when there are more than eight modulation modes, the performance of OFDMA-F$^2$L with the optimal selections does not change, as the total power required reaches the power budget $P_m^{\max}$.
When there are over three modulation modes, the performance of Baseline 1 stays the same since it only considers two modulation modes.
In contrast, the unstable performance of Baseline 2 results from the random client selection policy. These results corroborate the merits of OFDMA-F$^2$L with the optimal selections.

\section{Conclusion}\label{sec-con}

We have presented a new OFDMA-F$^2$L framework, which
allows clients to train varying numbers of iterations per aggregation round
% adapting to their channel conditions and computing powers.
and concurrently upload their local models,
thereby increasing participating clients and extending their training times.
A convergence upper bound on OFDMA-F$^2$L has been derived and minimized by 
jointly optimizing the selections of clients, subchannels, and modulations to maximize the weighted sum rate of the clients in each aggregation round of OFDMA-F$^2$L.
A Lagrange-dual based method has been developed to maximize the weighted sum rate, resulting in a ``winner-takes-all'' strategy delivering almost surely optimal selections of clients, subchannels, and modulations.
As tested experimentally on MLP and CNN,
OFDMA-F$^2$L with the optimal selections can efficiently reduce communication delay,
and improve the training convergence and accuracy, e.g., by about 18\% and 5\%, compared to its potential alternatives.

\appendices
\section{Proof of Theorem~\ref{theo_convergence bound}}\label{appendix_convergence bound}

Assuming that the gradient of the global loss function is $L_c$-Lipschitz continuous and the first-order Taylor expansion can be conducted to obtain the upper bound of
$F({\bm \omega}^{(\tau+1)}) -F({\bm \omega}^{\tau })$, as provided by
\begin{equation}\label{eq_taylor_expansion}
    \begin{aligned}
    F({\bm \omega}^{(\tau+1)}) -F({\bm \omega}^{\tau }) 
    & \leq \nabla F\left({\bm \omega}^{\tau } \right) \left({\bm \omega}^{(\tau + 1)} - {\bm \omega}^{\tau } \right) + \frac{L_c}{2} \left\| {\bm \omega}^{(\tau + 1)} - {\bm \omega}^{\tau} \right\|^2.
    \end{aligned}
\end{equation} 
By substituting~\eqref{eq_local_SGD} and~\eqref{eq_global_weight} into~\eqref{eq_taylor_expansion}, we have
\begin{subequations}\label{eq_F_w_t1}
		\begin{align}
	   F \left( {\bm \omega}^{(\tau+1)} \right) - F\left( {\bm \omega}^{\tau} \right) 
	   & \leq  -  \nabla F\left({\bm \omega}^{\tau} \right) \sum_{m \in {\cal M}} \varrho^{\tau}_m \zeta^{\tau}_m \left(\sum^{I^{\tau}_m -1}_{i=0}\eta \nabla F_m ({\bm{\omega}}^{\tau}_m (i)) \right)  \notag \\
   & \quad +  \frac{L_c}{2} \bigg\|- \sum_{m \in {\cal M}} \varrho^{\tau}_m \zeta^{\tau}_m \left(\sum^{I^{\tau}_m -1}_{i=0} \eta \nabla F_m ({\bm{\omega}}^{\tau}_m(i)) \right) \bigg\|^2\\
	 & =   \frac{\eta^2 L_c}{2} \bigg\|\sum_{m \in {\cal M}} \varrho^{\tau}_m \zeta^{\tau}_m \left(\sum^{I^{\tau}_m-1}_{i=0} \nabla F_m ({\bm{\omega}}^{\tau }_m(i)) \right) \bigg\|^2  \notag \\
  & \quad -  \eta \nabla F\left({\bm \omega}^{\tau } \right) \sum_{m \in {\cal M}} \varrho^{\tau}_m \zeta^{\tau}_m \left(\sum^{I^{\tau}_m-1}_{i=0} \nabla F_m ({\bm{\omega}}^{\tau}_m(i)) \right).
		\end{align}
\end{subequations}

We proceed by computing the expected value of both sides of \eqref{eq_F_w_t1} regarding ${\boldsymbol{{\cal B}}}^{\tau}$, representing the sampled mini-batches from $\cal D$ at the $\tau$-th global aggregation. It follows that
\begin{equation}\label{eq_F_w_t1_avg}
    \begin{aligned}
   \mathbb{E}_{{\boldsymbol{{\cal B}}}^{\tau}}\!\left\lbrace\! F\left({\bm \omega}^{(\tau + 1)} \right) \right\rbrace 
    & \leq  F\left({\bm \omega}^{\tau} \right) - \eta \left\|\nabla F\left({\bm \omega}^{\tau }\right) \right\|^2 \\
     & + \frac{\eta^2 L_c}{2} \mathbb{E}_{{\boldsymbol{{\cal B}}}^{\tau}}\!\left\lbrace\!\left\|\! \sum_{m \in {\cal M}} \varrho^{\tau}_m \zeta^{\tau}_m \left(\sum^{I^{\tau}_m -1}_{i=0} \nabla F_m ({\bm{\omega}}^{\tau}_m(i); {\cal B}^{\tau }_m(i)) \right) \right\|^2 \!\right\rbrace\!.
    \end{aligned}
\end{equation}
Since $\varrho^{\tau}_m = \frac{A}{I^{\tau}_m} \cdot \frac{D_m}{D}$, $\mathbb{E}_{{\boldsymbol{{\cal B}}}^{\tau}}\Big\| \sum_{m \in {\cal M}} \varrho^{\tau}_m  \left(\sum^{I^{\tau}_m-1}_{i=0} \nabla F_m ({\bm{\omega}}^{\tau }_m(i); {\cal B}^{\tau }_m(i))\right) \Big\|^2$ in \eqref{eq_F_w_t1_avg} is rewritten as
% \begin{figure*}[ht]
\begin{subequations}\label{eq_delta_Fm}
	\begin{align}
	 & \mathbb{E}_{{\boldsymbol{{\cal B}}}^{\tau}}\left\lbrace \left\| \sum_{m \in {\cal M}} \frac{A D_m \zeta^{\tau}_m}{D I^{\tau}_m} \left(\sum^{I^{\tau}_m -1}_{i=0} \nabla F_m ({\bm{\omega}}^{\tau }_m(i);{\cal B}^{\tau}_m(i)) \right) \right\|^2 \right\rbrace \\ \label{eq_delta_Fm a}
	 & = \sum_{m \in {\cal M}} \frac{A^2 D_m^2 \zeta^{\tau}_m}{D^2} \mathbb{E}_{{{\cal B}}^{\tau}_m}\left\lbrace \left\| \frac{1}{I^{\tau}_m} \sum^{I^{\tau}_m-1}_{i=0} \nabla F_m ({\bm{\omega}}^{\tau}_m(i);{\cal B}^{\tau}_m(i)) \right\|^2 \right\rbrace\\ \label{eq_delta_Fm b}
	& = \sum_{m \in {\cal M}}  \frac{A^2 D_m^2 \zeta^{\tau}_m}{D^2\left( I^{\tau}_m \right)^2} \sum^{I^{\tau}_m-1}_{i=0}  \mathbb{E}_{{{\cal B}}^{\tau}_m(i)} \left\|\nabla F_m ({\bm{\omega}}^{\tau}_m(i);{\cal B}^{\tau}_m(i)) \right\|^2 \\ \label{eq_delta_Fm c}
	& \leq \sum_{m \in {\cal M}}  \frac{A^2 D_m^2 \zeta^{\tau}_m}{D^2 I^{\tau}_m} \left(\kappa_1 + \kappa_2 \|\nabla F(\bm{\omega}^{\tau})\|^2 \right),
	\end{align}
\end{subequations}
where \eqref{eq_delta_Fm a} is due to the i.i.d. datasets of all clients; \eqref{eq_delta_Fm b} is because the sampled mini-batches are considered to be i.i.d. between the two consecutive aggregations; and \eqref{eq_delta_Fm c} is due to the assumption of $\mathbb{E}_{{\cal B}^{\tau}_m}\left\| \nabla F_m ({\bm{\omega}}^{\tau }_m(i); {\cal B}^{\tau }_m(i)) \right\|^2  \leq \kappa_1 + \kappa_2 \|\nabla F(\bm{\omega}^{\tau })\|^2, \forall m~\text{and}~\tau$.

Since $\mathbb{E}_{\boldsymbol{{\cal B}}^{\tau}}\lbrace F({\bm \omega}^{\tau + 1} ) \rbrace =F({\bm \omega}^{\tau + 1}) $, we substitute \eqref{eq_delta_Fm} into~\eqref{eq_F_w_t1_avg}, yielding 
\begin{subequations}\label{eq_F_w1 2}
	\begin{align}
	%& \mathbb{E}_{{\cal U}^t}\left\lbrace F\left({\bm{\omega}}^{t+1} \right)  \right\rbrace 
	 F({\bm \omega}^{(\tau + 1)}) & \leq F({\bm \omega}^{\tau}) - \phi_2 \|\nabla F({\bm \omega}^{\tau }) \|^2 + \phi_1, \tag{\ref{eq_F_w1 2}}
%	& F\left({\bm{\omega}}^{t} \right) + \frac{L}{2} \mathbb{E}_{{\cal U}^t}\left\lbrace\left\|{\bm n}^{t+1} \right\| \right\rbrace.
	\end{align}
\end{subequations}
where $\phi_1 \triangleq \frac{\eta^2 L_c}{2} \sum_{m \in {\cal M}}  \frac{A^2 \kappa_1 D_m^2 \zeta^{\tau}_m}{D^2 I^{\tau}_m}$ and $\phi_2 \triangleq \eta - \frac{\eta^2 L_c}{2} \sum_{m \in {\cal M}}  \frac{A^2 \kappa_2 D_m^2 \zeta^{\tau}_m}{D^2 I^{\tau}_m}$.

% Given $\mathbb{E}_{{\cal K}^m}\lbrace F({\bm \omega'}(t)) \rbrace = F({\bm \omega'}(t) ) $, 
Detracting $F({\bm \omega}^*)$ from both sides of \eqref{eq_F_w1 2}, it follows%~\eqref{eq_Fw1_Fw}.
%\begin{figure*}[ht]
\begin{equation}\label{eq_Fw1_Fw}
	\begin{aligned}
	 F({\bm \omega}^{(\tau + 1)} ) - F({\bm \omega}^*)
	 \leq  F({\bm \omega}^{\tau } ) - F({\bm \omega}^*) - \phi_2 \|\nabla F({\bm \omega}^{\tau }) \|^2 + \phi_1.
	\end{aligned}
\end{equation}
%\hrulefill
%\end{figure*}
Considering Polyak-Lojasiewicz condition,
we have
\begin{equation}\label{eq_delta_Fw}
     \|\nabla F({\bm \omega}^{\tau } ) \|^2 \geq 2 \rho \left( F(\bm{\omega}^{\tau }) - F(\bm{\omega}^{\ast})\right).
\end{equation}
By substituting \eqref{eq_delta_Fw} into \eqref{eq_Fw1_Fw}, we have
\begin{equation}\label{eq-recurrence}
	\begin{aligned}
	 F({\bm \omega}^{(\tau + 1)} ) - F({\bm \omega}^*)
	 \leq  (1 -  2\rho \phi_2)\left[F({\bm \omega}^{\tau } ) - F({\bm \omega}^*)\right] + \phi_1.
	\end{aligned}
\end{equation}
Based on~\eqref{eq-recurrence}, we can use mathematical induction and finally obtain
\begin{equation}
\begin{aligned}
F({\bm \omega}^{\tau } ) - F({\bm \omega}^*) & \leq \left[ F({\bm \omega}^0) - F({\bm \omega}^*)\right] \left(1 - 2\rho \phi_2\right)^{\tau} \\
& \quad + \phi_1 \frac{1-\left(1 - 2\rho \phi_2\right)^{\tau}}{ 2\rho \phi_2},
\end{aligned}
\end{equation}
which concludes this proof.

\ifCLASSOPTIONcaptionsoff
\newpage
\fi

\bibliographystyle{IEEEtran}
\bibliography{FL_Ref}

\end{document}